%% file: main.tex
\documentclass[aps,pra,notitlepage,onecolumn,nofootinbib,superscriptaddress,longbibliography,twocolumn,10pt]{revtex4-1}
\input{pretex}
\usepackage{times}
\usepackage{amsmath,amsfonts,amssymb,graphicx,mathtools,bm,tcolorbox,relsize}
\usepackage[utf8]{inputenc}
\usepackage[T1]{fontenc}
\usepackage[qm]{qcircuit}
\usepackage{braket}
\usepackage{natbib}
\usepackage{multirow}
\usepackage{booktabs}
\usepackage{tabularx}
\usepackage{array}
\usepackage{optidef}

\pgfplotsset{compat=1.16}
\usepackage[normalem]{ulem}
\newtheorem{proposition}{Proposition}

\newtheorem{theorem}{Theorem}
\newtheorem{corollary}{Corollary}

\newcommand{\idop}{I} %identity operator
\newcommand{\lrp}[1]{\left( #1 \right)}
\newcommand{\lrb}[1]{\left[ #1 \right]}
\newcommand{\lrc}[1]{\left\{ #1 \right\}}

\newcommand{\lrV}[1]{\left\| #1 \right\|}
\newcommand{\lrang}[1]{\left\langle #1 \right\rangle}

\usepackage[scr=boondox]{mathalpha}
\newcommand{\mS}{\mathscr{S}}

\newcommand{\mW}{\mathscr{W}}
\newcommand{\tp}{{\rm T}}

\allowdisplaybreaks

\begin{document}
\title{
The communication power of indefinite causal order
}

\author{Xuanqiang Zhao}
\affiliation{QICI Quantum Information and Computation Initiative, School of Computing and Data Science, The University of Hong Kong, Pokfulam Road, Hong Kong}

\author{Benchi Zhao}
\affiliation{QICI Quantum Information and Computation Initiative, School of Computing and Data Science, The University of Hong Kong, Pokfulam Road, Hong Kong}

\author{Cyril Branciard}
\affiliation{Universit\'e Grenoble Alpes, CNRS, Grenoble INP\footnote{Institute of Engineering Univ.\ Grenoble Alpes}, Institut N\'eel, 38000 Grenoble, France}

\author{Giulio Chiribella}
\affiliation{QICI Quantum Information and Computation Initiative, School of Computing and Data Science, The University of Hong Kong, Pokfulam Road, Hong Kong}
\affiliation{Quantum Group, Department of Computer Science, University of Oxford, Wolfson Building, Parks Road, Oxford, OX1 3QD, United Kingdom}
\affiliation{Perimeter Institute for Theoretical Physics, 31 Caroline Street North, Waterloo, Ontario, Canada}

%%%%%%%%%%%%%%%%%%%%%%%%%%%%%%%%%%%%%%%%%%%%%%%%%%%%%%%%%%%%
%%%%%%%%%%%%%%%%%%%%%%%%%%%%%%%%%%%%%%%%%%%%%%%%%%%%%%%%%%%%
\begin{abstract}
Quantum theory is in principle compatible with scenarios where physical processes occur in an indefinite order, potentially yielding advantages in a broad range of information processing tasks.
However, advantages in communication, the most basic form of information processing, have so far remained controversial and hard to prove. Here we provide a framework for assessing the role of causal order in communication, by comparing different causal structures under the constraint that the allowed operations must not generate signaling from signaling-incapable devices.
Using this framework, we establish a clear-cut advantage of indefinite causal order, and, at the same time, we identify a series of fundamental limits to the communication power of causal structures in quantum mechanics.
Notably, we find that a special form of indefinite causal order, obtained by coherently controlling the order of two processes, enhances the transmission of classical messages in a one-shot scenario, but no quantum operation with indefinite order can offer advantages over shared entanglement when asymptotically many uses of the same communication device are employed.
Overall, our results unveil non-trivial relations between communication, causal order, entanglement, and no-signaling quantum processes.
\end{abstract}

\date{\today}
\maketitle

%%%%%%%%%%%%%%%%%%%%%%%%%%%%%%%%%%%%%%%%%%%%%%%%%%%%%%%%%%%%
%%%%%%%%%%%%%%%%%%%%%%%%%%%%%%%%%%%%%%%%%%%%%%%%%%%%%%%%%%%%
\section{Introduction}
Traditionally, physical theories picture a world of processes taking place in a well-defined causal order.
It has been observed, however, that 
quantum theory is in principle compatible with scenarios where the order becomes indefinite~\cite{chiribella2009beyond,oreshkov2012quantum,chiribella2013quantum,bisio2018higher}.
Some of these scenarios can be engineered by letting a quantum system control the order in which two or more processes take place~\cite{chiribella2009beyond,wechs2021quantum,purves2021quantum}, a situation that has been reproduced in a series of photonic experiments~\cite{procopio2015experimental,rubino2017experimental,goswami2018indefinite,wei2019experimental,guo2020experimental,rubino2020experimental,goswami2020experiments,yin2023experimental,cao2023semi,rozema2024experimental}.
Other scenarios could be realized by delocalizing quantum systems in time~\cite{oreshkov2019time}, whereas the most general forms of indefinite order can be more exotic~\cite{oreshkov2012quantum,baumeler2014perfect,branciard2015simplest} and their physical implementation still remains an open question.
In general, scenarios with indefinite causal order are expected to offer new insights into the long-standing problem of quantizing gravity~\cite{hardy2007towards,hardy2009quantum}, and, at the same time, to provide a new angle on the characterization of quantum theory in terms of its information processing capabilities~\cite{fuchs2003quantum,brassard2005information,chiribella2016quantum}.

Over the past two decades, advantages of indefinite order have been discovered in a variety of information processing tasks, including quantum process discrimination~\cite{chiribella2012perfect,bavaresco2021strict, wechs2021quantum,bavaresco2022unitary}, query complexity~\cite{abbott2024quantum,abbott2025classical}, communication complexity~\cite{feix2015quantum,guerin2016exponential}, quantum metrology~\cite{zhao2020quantum, liu2023optimal,mothe2024reassessing}, and inversion of an unknown quantum dynamics~\cite{quintino2019reversing, quintino2019probabilistic, quintino2022deterministic}.
In all these tasks, quantum operations with indefinite order were proven to outperform all possible operations with definite order.
Quite strikingly, however, no such separation was found in the most basic information processing task, namely the transmission of messages from a sender to a receiver~\cite{shannon1948mathematical, shannon1956zero}.

A body of works~\cite{ebler2018enhanced,salek2018quantum,procopio2019communication,procopio2020sending,loizeau2020channel,caleffi2020quantum,goswami2020increasing,chiribella2021indefinite,chiribella2021quantum,chandra2021entanglement,sazim2021classical,bhattacharya2021random} showed that combining noisy channels in an indefinite order leads to enhancements over the standard communication scenario in which channels are combined in parallel or in sequence, without introducing additional systems that control their configuration.
However, other works ~\cite{abbott2020communication,guerin2019communication} observed that the above enhancements do not establish a fundamental separation between definite and indefinite order, but rather an advantage of indefinite order over a restricted class of scenarios with definite order.
Ref.~\cite{kristjansson2020resource} pointed out that restrictions on the allowed operations are unavoidable in the context of communication and proposed that the search for communication advantages of indefinite causal order should be framed in a resource-theoretic framework.
However, a unified framework treating definite and indefinite causal order on the same footing has been missing so far, and, as a consequence, the problem of assessing the communication power of indefinite order has remained open.

Here, we provide the first rigorous demonstration of a communication advantage of indefinite causal order, and, at the same time, we establish a number of fundamental limits to its communication power.
To compare operations with different causal structure, we develop a resource-theoretic framework where the devices capable of transmitting signals are considered resources, and the allowed operations do not generate signaling from signaling-incapable devices.
In this framework, we show that allowed operations with indefinite order enable the perfect, error-free transmission of a bit by combining two noisy channels that would inevitably cause errors if combined by any allowed operation with definite order. 
Remarkably, we find that this communication advantage can be achieved through quantum control over the order of the two channels, and therefore can in principle be demonstrated on a photonic platform.

In stark contrast, we show that there exists a broad set of
communication scenarios where indefinite causal order does
not offer advantages.
In particular, we find out that, strikingly, indefinite causal order does not offer any advantage over shared entanglement in the communication of classical and quantum information through asymptotically many uses of the same channel.
Overall, our results reveal that the relations between communication, causal order, entanglement, and no-signaling operations are more complex, and more profound than previously envisaged.

%%%%%%%%%%%%%%%%%%%%%%%%%%%%%%%%%%%%%%%%%%%%%%%%%%%%%%%%%%%%
%%%%%%%%%%%%%%%%%%%%%%%%%%%%%%%%%%%%%%%%%%%%%%%%%%%%%%%%%%%%
\section{Results}
%%%%%%%%%%%%%%%%%%%%%%%%%%%%%%%%%%%%%%%%%%%%%%%%%%%%%%%%%%%%
{\bf The resource theory of signaling.}
Here we develop a resource theory where the key resource is the ability to send signals.
In this resource theory, the basic objects are communication devices, and the basic operations assemble the devices into communication links connecting a sender to a receiver (see Fig.~\ref{fig:channel_assembly}).
Mathematically, the basic resources are represented by lists of quantum channels (completely positive, trace preserving maps transforming the density matrix of an input system into the density matrix of an output system).
For a list of $N$ channels, we will use the vector notation $\vec{\cC} = (\cC_1, \dots, \cC_N)$, and we will denote by $X_i$ ($Y_i$) the input (output) system of the $i$-th channel in the list.

The resource theory is constructed by first specifying a set of devices which are regarded as non-resourceful, or free, and then by considering the largest set of operations that do not create resources~\cite{chitambar2019quantum}.
Here, we focus on the resource of signaling: we regard as free all the devices that block any signal, by outputting a fixed state independently of their input.
Mathematically, these devices are described by replacement channels, that is, quantum channels $\cC$ satisfying the condition $\cC(\rho) = \rho_0 \, \forall \rho$, where $\rho$ is an arbitrary density matrix of the input system, and $\rho_0$ is a fixed density matrix of the output system.
In general, free resources are described by lists of replacement channels.

We now specify the allowed operations that can be used to assemble quantum channels into communication links.
Mathematically, operations on quantum channels are described by quantum supermaps~\cite{chiribella2008transforming,chiribella2008quantum,chiribella2013quantum,bisio2018higher}, that is, linear maps transforming (lists of) quantum channels into quantum channels.
In the following, we will focus on the case in which the $N$ input channels in the initial list are assembled to build a single communication channel connecting a sender to a receiver.
The input and output systems of the channel produced by the supermap will be denoted by $A$ and $B$, respectively.

\begin{figure}
    \centering
    \includegraphics[width=\columnwidth]{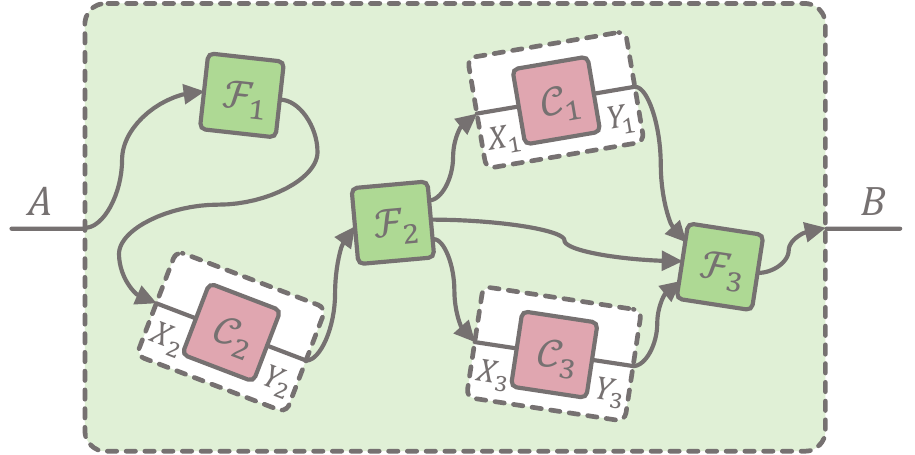}
    \caption{\textbf{Communication through multiple quantum devices connected into a network.}
    A set of $N$ communication devices, described by quantum channels $\vec{\cC} = (\cC_1,\dots,\cC_N)$ with input/output pairs $(X_i,Y_i)_{i=1}^N$, is used to build a communication link from a sender $A$ to a receiver $B$.
    In the example shown in this figure, $N=3$ channels are connected in a fixed order, interspersed by additional channels $\cF_1, \cF_2$, and $\cF_3$.
    The additional channels can be arbitrary, and possibly be connected by direct communication links (as in the case of channels $\cF_2$ and $\cF_3$ in the picture), provided that no overall signaling from $A$ to $B$ is generated when the channels $\cC_1$, $\cC_2$, and $\cC_3$ are replacement channels.
    In general, the initial devices can be assembled in any causal structure, including causal structures with indefinite order.}
    \label{fig:channel_assembly}
\end{figure}

The largest set of operations that do not generate resources from free devices is the set of signaling-non-generating supermaps, that is, quantum supermaps that transform lists of replacement channels into replacement channels.
These supermaps may internally use additional channels that aid the communication between sender and receiver, but must block every signal if the original channels $\vec{\cC}$ are replacement channels.
Since here we are focusing on the fundamental aspects of communication, we will take the full set of signaling-non-generating supermaps as the free operations in the resource theory of signaling.
In the special case where each list consists of a single channel, the set of signaling-non-generating supermaps has been known as ``no-signaling resources'' and has been fruitfully adopted in quantum Shannon theory~\cite{leung2015power,duan2016no,wang2016quantum,wang2018separation}.
Signaling-non-generating supermaps have been recently used to formulate a resource-theoretic framework for quantum communication~\cite{fang2019quantum,takagi2020application}.
Our resource theory of signaling is the extension of this resource-theoretic framework to the case where multiple input channels can be assembled in a non-trivial order.

A complete characterization of the signaling-non-generating supermaps is provided in Methods, where we show that they are in one-to-one correspondence with a subset of multipartite channels satisfying a no-signaling condition: explicitly, a signaling-non-generating supermap $\mS$ transforming a list of $N$ channels with input/output pairs $(X_i, Y_i)_{i=1}^N$ into a channel with input/output pair $(A, B)$ corresponds to a multipartite channel with input system $AY$ and output system $BX$ (with the notation $X \coloneqq X_1\cdots X_N$ and $Y \coloneqq Y_1\cdots Y_N$), satisfying the condition that no signal can be sent from $A$ to $B$.
We call this condition the ``No Forward Signaling Condition.''
Informally, the No Forward Signaling Condition states that the supermap does not allow signals to travel from the sender to the receiver without the help of the devices that are used to build the communication link.
This condition guarantees that our theory satisfies a condition proposed by Ref.~\cite{kristjansson2020resource} as a minimum requirement for any resource theory of communication.

With the set of free operations in place, we are now in position to rigorously analyze the role of indefinite causal order.
For this purpose, we define three subsets of the set of free operations, ordered by mutual inclusion.
The smallest set is the set of free supermaps that can be achieved by placing the $N$ input channels in parallel inside a quantum circuit~\cite{chiribella2008transforming}.
The second set is the set of free supermaps that can be achieved by placing the $N$ input channels inside a quantum circuit in a fixed, pre-assigned order.
These supermaps are known as quantum combs~\cite{chiribella2008quantum,chiribella2009theoretical}.
The third set is the largest set of free supermaps that can be achieved with a definite causal order.
These supermaps, often called ``causally separable'' supermaps~\cite{oreshkov2016causal,wechs2019on}, do not require the order to be fixed {\em a priori}, but rather allow for it to be dynamically determined at subsequent steps of a quantum circuit.
In the following, the sets of free supermaps with parallel placement, fixed-order placement, and definite (possibly dynamical) order placement will be denoted by ${\bf FreePar}$, ${\bf FreeFix}$, and ${\bf FreeDef}$, respectively.
A characterization of these sets of supermaps is provided in Supplementary Note~\ref{supp_sec:causal_order}.
Notice that one has the (generally strict) inclusions ${\bf FreePar} \subset {\bf FreeFix} \subset {\bf FreeDef} \subset {\bf Free}$, where ${\bf Free}$ denotes the set of all supermaps that are free in our resource theory.
Supermaps that are in ${\bf Free}$ but not in ${\bf FreeDef}$ are free operations with indefinite causal order.
In the following, we will show that using free operations with indefinite causal order offers an advantage in one of the most fundamental communication tasks: the transmission of bits from a sender to a receiver.

\medskip

%%%%%%%%%%%%%%%%%%%%%%%%%%%%%%%%%%%%%%%%%%%%%%%%%%%%%%%%%%%%
{\bf Communication advantages of indefinite causal order.}
Suppose that Alice wants to transmit classical messages to Bob, using a communication link built from $N$ communication devices using signaling-non-generating operations.
We now show that assembling the devices in an indefinite order can increase the number of bits that Alice can reliably send to Bob.

Let us start from the ideal scenario in which the communication is free from errors.
In Methods, we show that the task of zero-error classical communication can be formulated as the task of converting $N$ communication devices into a noiseless classical channel.
Mathematically, the conversion is implemented by a supermap $\mS$ that transforms a list of $N$ quantum channels $\vec{\cC}$ into a noiseless classical channel $\Delta_m$, of the form $\Delta_m: \rho \mapsto \sum_{j=0}^{m-1} \bra{j}\rho\ket{j}\,\proj{j}$, where $m$ is the number of classical messages the channel can send and $\lrc{\ket{j}}_{j=0}^{m-1}$ is a fixed orthonormal basis.

Now, the question is: what is the maximum number of messages Alice can communicate to Bob in an error-free way?
When the $N$ communication devices are assembled using operations in the set ${\bf X}$ (where ${\bf X}$ can be ${\bf FreePar}$, or ${\bf FreeFix}$, or ${\bf FreeDef}$, or ${\bf Free}$), the maximum number of messages is $m_0^{\bf X} (\,\vec{\cC}\,) \coloneqq \max \{m \,|\, \exists \mS \in {\bf X}:\, \mS (\,\vec{\cC}\,) = \Delta_m\}$.
Due to the inclusions ${\bf FreePar} \subset {\bf FreeFix} \subset {\bf FreeDef} \subset {\bf Free}$, the inequalities $m_0^{\bf FreePar} (\,\vec{\cC}\,) \le m_0^{\bf FreeFix} (\,\vec{\cC}\,) \le m_0^{\bf FreeDef} (\,\vec{\cC}\,) \le m_0^{\bf Free} (\,\vec{\cC}\,)$ hold for every list of channels $\vec {\cC}$.
We now show that the last of these inequalities is generally strict: assembling the devices in an indefinite causal order generally increases the number of messages Alice can communicate to Bob in an error-free way.

Let us introduce the notion of one-shot zero-error capacity of the channels $\vec{\cC}$ assisted by signaling-non-generating operations in the set $\bf X$, defined as $C^{\bf X}_0 (\,\vec{\cC}\,) \coloneqq \log_2 m_0^{\bf X} (\,\vec{\cC}\,)$.
Then, consider the case of $N=2$ amplitude damping channels, that is, qubit channels of the form $\cA: \rho \mapsto \cA(\rho) = A_0 \rho A_0^\dag + A_1 \rho A_1^\dag$, where $\rho$ is an arbitrary density matrix of a single qubit, $A_0$ and $A_1$ are the 2-by-2 matrices given by $A_0 \coloneqq \sqrt{\eta} \ketbra{0}{1}$ and $A_1 \coloneqq \sqrt{I - A_0^\dag A_0}$, and $\eta \in [0,1]$ is the damping parameter.
We then evaluate the zero-error capacity, setting $N=2$ and $\cC_1 = \cC_2 = \cA$.

In Methods and Supplementary Note~\ref{supp_sec:capacity}, we show that the evaluation of the capacity can be reduced to a semidefinite program. Using the semidefinite programs for $C_0^{\bf FreeDef}$ and $C_0^{\bf Free}$, we can then pin down the advantage of indefinite causal order, both analytically and numerically.
On the analytical side, we prove that for damping parameter $\eta = 0.1$, the zero-error capacity assisted by signaling-non-generating operations with definite causal order is
\begin{align}\label{freedef=0}
    C_0^{\bf FreeDef}(\,\vec{\cC}\,) = 0\,,
\end{align}
meaning that no information can be communicated without error when the two channels are combined in a definite order subject to the condition of signaling-non-generation.
In stark contrast, we find that allowing signaling-non-generating operations with indefinite causal order yields capacity
\begin{align}
    C_0^{\bf Free}(\,\vec{\cC}\,) = 1\,,
\end{align}
meaning that indefinite causal order enables the error-free communication of one bit while still complying with the constraint of signaling-non-generation.

Remarkably, we find that the above communication advantage can be achieved by a restricted class of operations with indefinite causal order: in Supplementary Note~\ref{supp_sec:advantage}, we show that the supermap achieving 1 bit of classical communication can be realized by a quantum circuit with quantum control on the order (QC-QC)~\cite{wechs2021quantum,purves2021quantum}.
This finding is important because QC-QC operations can in principle be reproduced in photonic experiments.
More specifically, we show that a circuit achieving the advantage can be built from the quantum switch~\cite{chiribella2009beyond,chiribella2013quantum}, a supermap that applies two channels ${\cal C}_1$ and ${\cal C}_2$ in an order determined by the state of a quantum bit.
In its original formulation, the quantum switch allows for an intermediate operation to be performed between channels ${\cal C}_1$ and ${\cal C}_2$, although in previous applications this intermediate operation was generally taken to be trivial.
In contrast, achieving the communication advantage here requires a non-trivial intermediate operation, acting not only on the qubit passing through channels ${\cal C}_1$ and ${\cal C}_2$, but also on auxiliary systems, as specified in Supplementary Note~\ref{supp_sec:advantage}.
In this respect, our results provide the first application of the full-fledged quantum switch, placing two channels in an indefinite order and allowing for general intermediate operations to take place between them.

Numerically, we determine the values of the zero-error capacity for various values of the damping parameter $\eta$ between $0$ and $0.5$, finding that the advantage of indefinite causal order is present for all tested values of $\eta$ in the interval between $\eta = 0$ (excluded) and $\eta \approx 0.2$, as illustrated in Fig.~\ref{fig:capacities}.
Together, these results provide the first rigorous evidence of a communication advantage of indefinite causal order.

A similar approach can be used to derive advantages of indefinite causal order in classical communication beyond the zero-error scenario.
The maximum number of messages that can be transmitted with error upper bounded by $\epsilon$, using the channels $\vec{\cC}$ assembled by operations in ${\bf X}$, is given by $m_\epsilon^{\bf X}(\,\vec{\cC}\,) \coloneqq \max \{m \,|\, \exists \mS \in {\bf X}:\, \|\mS(\,\vec{\cC}\,) - \Delta_m\|_\diamond / 2 \le \epsilon\}$, where $\|\cdot\|_\diamond$ denotes the diamond norm~\cite{kitaev1997quantum}.
We then introduce the one-shot capacity $C_\epsilon^{\bf X}(\,\vec{\cC}\,) \coloneqq \log_2 m_\epsilon^{\bf X}(\,\vec{\cC}\,)$, and show that there exist parameter values for which the maximum one-shot capacity achieved by signaling-non-generating operations with definite causal order is strictly smaller than the one-shot capacity achievable by signaling-non-generating operations with indefinite causal order.
For $N=2$ amplitude damping channels, an illustration of this fact is provided in Fig.~\ref{fig:capacities}.

\begin{figure}
    \centering
    \includegraphics[width=\columnwidth]{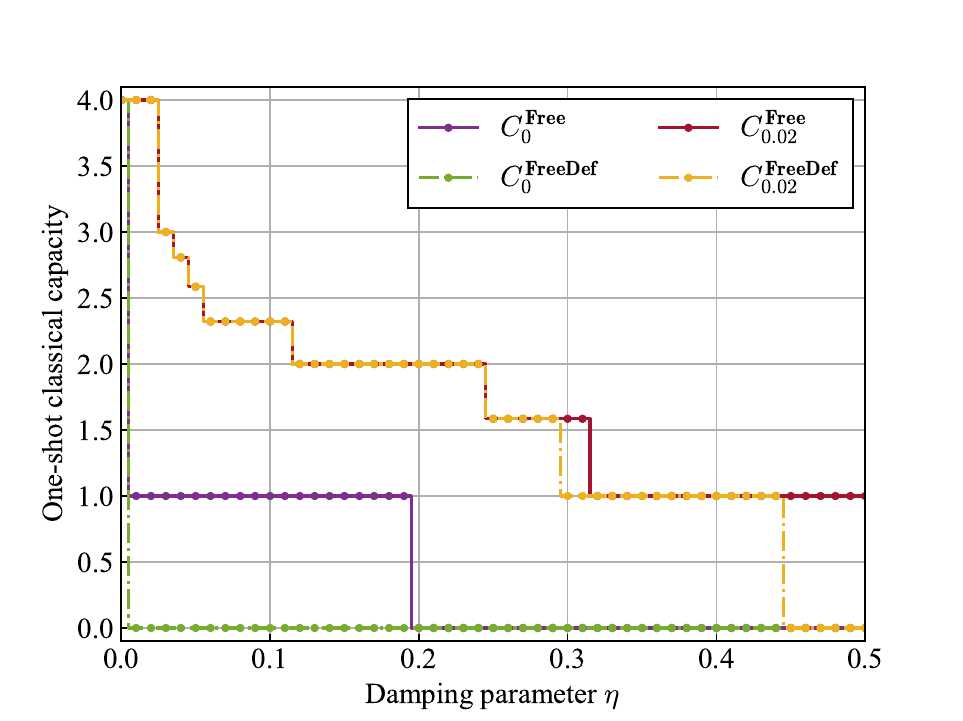}
    \caption{\textbf{Advantages of indefinite causal order in zero-error and bounded-error communication.}
    We compare the one-shot classical capacities of a pair of amplitude damping channels assisted by operations in the sets ${\bf Free}$ and ${\bf FreeDef}$, for zero-error and for bounded-error with error tolerance of $\epsilon = 0.02$.}
    \label{fig:capacities}
\end{figure}

\medskip

%%%%%%%%%%%%%%%%%%%%%%%%%%%%%%%%%%%%%%%%%%%%%%%%%%%%%%%%%%%%
{\bf No advantage for Pauli channels.}
After having established a rigorous advantage of indefinite causal order, we now show that the communication power of indefinite causal order is also subject to fundamental limitations.
In this section, we start by showing that no classical communication advantage is present for a large class of channels known as Pauli channels~\cite{wilde2013quantum}.

A Pauli channel acting on $q$ qubits is a channel of the form $\cP(\rho) = \sum_{i_1 \dots i_q} \, p_{i_1\dots i_q} \, (\sigma_{i_1} \ox \cdots \ox \sigma_{i_q}) \,\rho\, (\sigma_{i_1} \ox \cdots \ox \sigma_{i_q} )$, where $\sigma_0$ is the identity matrix, $\sigma_1, \sigma_2, \sigma_3$ are the three Pauli matrices, and $(p_{i_1\dots i_q})$ is a probability distribution, providing the probability that the $q$ qubits undergo the product gate $\sigma_{i_1} \ox \cdots \ox \sigma_{i_q}$.
Here we consider the situation where the original communication devices are $N$ Pauli channels $\vec{\cP} \coloneqq (\cP_1,\dots,\cP_N)$.
In general, these $N$ channels can be different from each other and can even act on different numbers of qubits.

In Supplementary Note~\ref{supp_sec:pauli}, we prove that indefinite causal order does not enhance the classical capacity, even in the presence of an error tolerance.
In fact, we prove an even stronger result: placing the $N$ channels in parallel with the assistance of signaling-non-generating operations (which in the parallel case, as in the single channel case, are known as no-signaling resources~\cite{leung2015power,duan2016no,wang2016quantum,wang2018separation}) resources is already sufficient to achieve the largest possible value of the capacity; in formula, $C_\epsilon^{\bf FreePar}(\,\vec{\cP}\,) = C_\epsilon^{\bf Free}(\,\vec{\cP}\,)$ for every $\epsilon \ge 0$.
In other words, when no-signaling resources are available, neither definite nor indefinite causal order can enhance the capacity beyond the maximum value achievable by placing the channels in parallel.

\medskip

%%%%%%%%%%%%%%%%%%%%%%%%%%%%%%%%%%%%%%%%%%%%%%%%%%%%%%%%%%%%
{\bf No advantage for channel simulation.}
We now establish another limitation of indefinite causal order, showing that it does not offer any advantage over no-signaling resources in the task of simulating a noisy quantum channel through classical communication~\cite{bennett2002entanglement,fang2019quantum}.
The task of noisy channel simulation is the inverse of the classical communication task: while in communication one uses a noisy channel to simulate a noiseless transmission of bits, in simulation one uses a noiseless transmission of bits to simulate a noisy channel.

In our resource-theoretic framework, the task of channel simulation can be phrased as the problem of converting a set of ideal classical channels into a target channel $\cC$ by means of signaling-non-generating operations.
Let us denote the list of ideal classical channels as $\vec{\Delta}_m = (\Delta_{m_1}, \Delta_{m_2}, \dots)$, where $m_i$ is the number of messages transmitted by the $i$-th channel, and $m \coloneqq \prod_j m_j$ is the total number of messages transmitted.
With this notation, the classical simulation cost for one realization of channel $\cC$ assisted by signaling-non-generating operations in the set ${\bf X}$ is
$S_{\epsilon}^{\bf X}(\cC) \coloneqq \min \{ \log_2 m \,|\, \exists \mS \in {\bf X}:\, \|\mS(\vec{\Delta}_{m}) - \cC\|_\diamond / 2 \le \epsilon\}$, which measures the minimum number of classical bits consumed by the simulation.
This type of one-shot cost is widely used in quantum resource theories~\cite{liu2019one,takagi2022one}, and has played an important role in the study of entanglement~\cite{buscemi2011entanglement,wang2020cost}, coherence~\cite{zhao2018one,diaz2018using}, and communication~\cite{duan2016no,fang2019quantum}.

Since enlarging the set of allowed operations can in principle reduce the simulation cost, one has the inequalities $S_{\epsilon}^{\bf FreePar}(\cC) \ge S_{\epsilon}^{\bf FreeFix}(\cC) \ge S_{\epsilon}^{\bf FreeDef}(\cC) \ge S_{\epsilon}^{\bf Free}(\cC)$ for every error tolerance $\epsilon \ge 0$.
Surprisingly, however, we found out that this chain of inequalities collapses to a single equality: $S_{\epsilon}^{\bf FreePar}(\cC) = S_{\epsilon}^{\bf Free}(\cC)$ for every $\epsilon \ge 0$ (see Supplementary Note~\ref{supp_sec:chan_sim} for the derivation).
In other words, signaling-non-generating operations with indefinite causal order do not offer any advantage over signaling-non-generating operations that put the original channels in parallel: all the classical communication can be sent in a single round.

\medskip

%%%%%%%%%%%%%%%%%%%%%%%%%%%%%%%%%%%%%%%%%%%%%%%%%%%%%%%%%%%%
{\bf No communication advantage in the asymptotic scenario.}
We now conclude our list of no-go results with a major limitation on the communication power of indefinite causal order: when the same quantum channel is used asymptotically many times, combining the uses in an indefinite causal order does not offer any advantage over the assistance of shared entanglement, both in classical and in quantum communication.

Let us start from the case of classical communication. In our resource-theoretic framework, the problem can be formalized as follows: Alice and Bob communicate through a link established by combining $N$ communication devices implementing the same quantum channel, namely $\vec{\cC}_{N} = (\cC_1, \dots, \cC_N)$ with $\cC_j = \cC$ for every $j \in \{1, \dots, N\}$.
Then, the classical capacity of the channel $\cC$ assisted by operations in the set ${\bf X}$ is defined as the maximum number of bits communicated reliably per channel use in the asymptotic limit of many uses and vanishing error, assembled by signaling-non-generating operations in ${\bf X}$.
In formula, $C^{\bf X} (\cC) \coloneqq \lim_{\epsilon \to 0}\lim_{N\to \infty} C_{\epsilon}^{\bf X} (\vec{\cC}_N)/N$.

We now show that the order of the channels, either definite or indefinite, does not affect the capacity in the asymptotic limit.
The proof is based on the quantum reverse Shannon theorem~\cite{bennett2014quantum,berta2011quantum}, which states the equality $C^{\bf Ent}(\cC) = S^{\bf Ent}(\cC)$, where $C^{\bf Ent}(\cC)$ and $S^{\bf Ent}(\cC)$ denote the entanglement-assisted classical capacity and the entanglement-assisted classical simulation cost of the channel $\cC$, respectively.

To prove our no-go result, we consider the asymptotic classical simulation cost assisted by supermaps in the set ${\bf X}$, defined as $S^{\bf X}(\cC) \coloneqq \lim_{\epsilon \to 0}\lim_{N \to \infty} S_{\epsilon}^{\bf X}(\vec{\cC}_N)$.
Since shared entanglement is a special case of a no-signaling resource~\cite{duan2016no,wang2016quantum,wang2018separation} (hence, also of a signaling-non-generating resource), we have the inequalities $C^{\bf Ent}(\cC) \leq C^{\bf FreeDef}(\cC)\leq C^{\bf Free}(\cC)$ and $S^{\bf Free}(\cC) \leq S^{\bf FreeDef}(\cC) \leq S^{\bf Ent}(\cC)$.
Moreover, the classical capacity of a channel cannot be larger than its simulation cost~\cite{fang2019quantum}, whence we have the inequality $C^{\bf Free}(\cC) \le S^{\bf Free}(\cC)$.
Combining the above inequality with the quantum reverse Shannon theorem, we then obtain the chain of equalities $C^{\bf Ent}(\cC) = C^{\bf Free}(\cC) = S^{\bf Free}(\cC) = S^{\bf Ent}(\cC)$.
This chain of equalities has two implications.
The first implication is that the interconversions of quantum channels achieved by signaling-non-generating supermaps (possibly including indefinite causal order) are reversible:
the cost to simulate a target channel equals the amount of resource that can be extracted from it.
This reversibility is desirable as it ensures no resources are wasted, and a single measure, in this case, the capacity, suffices to govern all asymptotic interconversions of channels~\cite{gour2024resources}.
The second implication is that neither definite nor indefinite causal order provide any advantage in the asymptotic regime: the assistance of entanglement is already sufficient to achieve the maximum possible classical capacity.

Note that the above argument also holds for the communication of quantum information.
By teleportation~\cite{bennett1993teleporting} and superdense coding~\cite{bennett1992communication}, both the entanglement-assisted quantum capacity $Q^{\bf Ent}$ and the entanglement-assisted quantum simulation cost $S_Q^{\bf Ent}$ are half of their classical counterparts: $Q^{\bf Ent} = \frac{1}{2}C^{\bf Ent} = \frac{1}{2}S^{\bf Ent} = S_Q^{\bf Ent}$.
Then, the same argument presented in the previous paragraph implies that neither definite nor indefinite causal order can enhance the quantum capacity in the asymptotic regime.

%%%%%%%%%%%%%%%%%%%%%%%%%%%%%%%%%%%%%%%%%%%%%%%%%%%%%%%%%%%%
%%%%%%%%%%%%%%%%%%%%%%%%%%%%%%%%%%%%%%%%%%%%%%%%%%%%%%%%%%%%
\section{Discussion}
It is instructive to compare our results with previous findings on quantum communication assisted by indefinite causal order.
Earlier works~\cite{ebler2018enhanced,salek2018quantum,procopio2019communication,procopio2020sending,loizeau2020channel,caleffi2020quantum,goswami2020increasing,chiribella2021indefinite,chiribella2021quantum,chandra2021entanglement,sazim2021classical,bhattacharya2021random} showed that the quantum switch~\cite{chiribella2009beyond,chiribella2013quantum}, a particular example of operation with indefinite causal order, can outperform the standard operations with definite causal order considered in quantum communication.
In those examples, the quantum switch enhanced both classical and quantum communication through Pauli channels, even in the asymptotic limit of large channel uses.
These findings are not in contradiction with our no-go theorems, for two major reasons.
First, the advantage of the quantum switch only applies to a restricted class of operations, not to the whole set of signaling-non-generating operations with indefinite causal order.
Second, and most important, the quantum switch is a signaling-generating operation: in a different language, this fact was observed in the first work on indefinite order in quantum communication~\cite{ebler2018enhanced}, where it was found that the quantum switch enables classical communication using two completely depolarizing channels, which are an example of replacement channels.
Since the quantum switch is signaling-generating, it is not an allowed operation in the resource theory of signaling considered in the present paper.

While the quantum switch is signaling-generating, combining it with other operations can still yield supermaps that satisfy the condition of no-signaling-generation.
As observed earlier in the paper, the quantum supermap that achieves the classical communication advantage can be realized in terms of the quantum switch ~\cite{chiribella2009beyond,chiribella2013quantum}, in its original version allowing for an non-trivial intermediate operation to take place between the two channels acting in an indefinite order.
In this realization, the quantum switch with a suitable choice of intermediate operation yields a new quantum channel, which is then used to store information about the initial noisy channels into a quantum state.
Eventually, this quantum state is used as a program to set up a communication link between the sender and receiver.
While the quantum switch alone can be signaling-generating, its combination with a suitable state preparation, a suitable intermediate operation, and a suitable programmable channel eventually gives a supermap that satisfies the condition of signaling-non-generation.
Interestingly, this storage-and-retrieval structure appears to be common to the optimal communication strategies with definite order, too, as observed in Supplementary Note~\ref{supp_sec:advantage}.
Overall, the advantage proven in this paper shows that indefinite causal order enables new ways to store quantum channels into quantum states, with provably different features compared to the standard way of storing quantum channels in a definite order~\cite{vidal2002storing,bisio2010optimal,sedlak2019optimal,yang2020optimal}.

In principle, the communication advantage shown in this paper could be demonstrated on a photonic platform, in a similar way as it was done in the quantum switch experiments~\cite{procopio2015experimental,rubino2017experimental,goswami2018indefinite,wei2019experimental,guo2020experimental,rubino2020experimental,goswami2020experiments,yin2023experimental,cao2023semi,rozema2024experimental}.
Compared to previous experiments, however, the current experiment presents new challenges.
First, the experiment would require controlling the order of two amplitude damping channels, which have not yet featured in previous experiments on indefinite causal order.
Second, demonstrating the advantage requires a non-trivial quantum channel (which also involves auxiliary systems) to be performed between the two channels put in an indefinite order.
Overall, our findings motivate further experimental efforts generate new and more complex quantum circuits with coherent quantum control over the order.

On the conceptual side, this work provides a rigorous resource-theoretic framework for investigating the communication advantages of indefinite causal order.
While different types of dynamical resources have been recently investigated in a series of works~\cite{ben2017resource,rosset2018resource,theurer2019quantifying,saxena2020dynamical,gour2020dynamical,takagi2020application,regula2021fundamental,kim2021one,milz2022resource,saxena2022quantifying,stratton2024dynamical,zhao2024probabilistic,zhu2025bidirectional}, the role of the causal structure within the set of free operations has so far remained unexplored.
In this respect, our framework provides a first exploration of the power of different causal structures in the conversion of dynamical resources.
An interesting direction for future research is to apply this approach to other resource theories, including the resource theory of quantum thermodynamics~\cite{felce2020quantum} and the resource theory of magic~\cite{mo2024enhancement}.

%%%%%%%%%%%%%%%%%%%%%%%%%%%%%%%%%%%%%%%%%%%%%%%%%%%%%%%%%%%%
%%%%%%%%%%%%%%%%%%%%%%%%%%%%%%%%%%%%%%%%%%%%%%%%%%%%%%%%%%%%
\section{Methods}
%%%%%%%%%%%%%%%%%%%%%%%%%%%%%%%%%%%%%%%%%%%%%%%%%%%%%%%%%%%%
{\bf Characterization of the free supermaps.}
Quantum supermaps are in one-to-one correspondence with subsets of multipartite quantum channels.
Specifically, a supermap mapping $N$ channels with input/output pairs $(X_i, Y_i)_{i=1}^N$ into a quantum channel with input $A$ and output $B$ corresponds to a multipartite quantum channel with input $YA$ and output $XB$, where $X \coloneqq X_1\cdots X_N$ and $Y \coloneqq Y_1\cdots Y_N$.
By the Choi isomorphism~\cite{choi1975completely}, every channel can be uniquely represented by a linear operator, known as its Choi operator.
For a multipartite channel corresponding to a supermap, the Choi operator also serves as a representation of the corresponding supermap.
The Choi operator of a supermap is sometimes called the ``process matrix''~\cite{oreshkov2012quantum}.

Supermaps with different causal structures correspond to multipartite channels satisfying different types of linear constraints.
In particular, the constraints on causally fixed, ({\em i.e.}, quantum combs), causally definite, and arbitrary supermaps were characterized in Ref.~\cite{chiribella2009theoretical}, Ref.~\cite{oreshkov2016causal,wechs2019on}, and Refs.~\cite{oreshkov2012quantum,oreshkov2016causal,araujo2015witnessing}, respectively.
The explicit expression of these constraints is reviewed in Supplementary Note~\ref{supp_sec:causal_order}, where we also provide an equivalent characterization of the general supermaps with $N$ input channels: a linear operator $J_{XYAB}$ is the Choi operator of a supermap with $N$ input channels if and only if
\begin{align}
    J_{XYAB} \geq 0 \quad\text{and}\quad \cL^{\rm NS}_{XY}\lrp{J_{XYA}} = \frac{\idop_{XYA}}{d_{X}},
\end{align}
where $J_{XYA} \coloneqq \tr_B[J_{XYAB}]$ (more generally, by dropping a subscript on an operator we mean tracing the corresponding subsystem out), $\cL^{\rm NS}_{XY}$ is the orthogonal projection onto the subspace of Hermitian operators spanned by the Choi operators of no-signaling multipartite quantum channels from $X \coloneqq X_1\cdots X_N$ to $Y \coloneqq Y_1\cdots Y_N$, $\idop_{XYA}$ is the identity operator on the spaces $XYA$, and $d_X$ is the dimension of system $X$.

Free supermaps in our resource theory must satisfy, in addition to the constraints associated to the causal structure, the constraint of signaling-non-generation.
In Supplementary Note~\ref{supp_sec:causal_order}, we prove that a supermap is signaling-non-generating if and only if the corresponding channel from $YA$ to $XB$ is no-signaling from $A$ to $B$.
In terms of the Choi operator $J_{XYAB}$, this condition reads \cite{leung2015power}:
\begin{align}\label{eq:fns_constraint}
    J_{YAB} = J_{YB} \ox \frac{\idop_A}{d_A}
\end{align}
where again $\idop_A$ is the identity operator on system $A$, and $d_A$ is the dimension of the corresponding Hilbert space. We call this constraint the No Forward Signaling Condition.

\medskip 

%%%%%%%%%%%%%%%%%%%%%%%%%%%%%%%%%%%%%%%%%%%%%%%%%%%%%%%%%%%%
{\bf One-shot classical capacity assisted by different types of causal structures.}
Here we show that the evaluation of the one-shot classical capacity of a list of quantum channels $\vec{\cC}$, assisted by various types of operations with definite and indefinite causal order can be reduced to a semidefinite program (SDP).
More generally, we consider the task of one-shot classical communication assisted by operations in a generic set of supermaps $\bf X$, with the only requirement that $\bf X$ is convex and closed under the local application of quantum channels at the sender's and receiver's ends.
Mathematically, the second condition is that, for every supermap $\mS$ in $\bf X$ with systems $A$ and $B$ at the sender's and receiver's end, respectively, for every pair of systems $A'$ and $B'$, and for every encoding channel $\cE$ (decoding channel $\cD$) transforming system $A'$ into system $A$ (system $B$ into system $B'$), the supermap $\mS_{\cE, \cD}$ defined by $\mS_{\cE,\cD} (\,\vec{\cC}\,) \coloneqq \cD\circ \mS(\,\vec{\cC}\,) \circ \cE$ is also in $\bf X$.
A consequence of this condition is that the set $\bf X$ contains supermaps with input systems $A'$ and output systems $B'$ of all possible dimensions.

The general settings are as follows.
An initial set of quantum devices, described by a set of quantum channels $\vec{\cC}$, is assembled into a communication link described by channel $\mS(\,\vec{\cC}\,)$, by applying a supermap $\mS$ in the set $\bf X$.
We denote by ${\bf X}_{AB}$ the set of supermaps in the set $\bf X$ that have systems $A$ and $B$ at the sender's and receiver's ends, respectively.
For the communication of $m$ messages, the channel $\mS(\,\vec{\cC}\,)$ is sandwiched between an encoding channel $\cE_m$, performed at the sender's side, and a decoding channel $\cD_m$, performed at the receiver's side.
Here, the encoding channel $\cE_m$ transforms an $m$-dimensional classical system $A_m$ into system $A$, while the decoding channel $\cD_m$ transforms system $B$ into an $m$-dimensional classical system $B_m$.

On average over all possible $m$ messages, the probability of error of the above scheme is $p(\vec{\cC},\mS ,m,\cE_m, \cD_m) \coloneqq 1 - \frac{1}{m}\sum_{j=0}^{m-1} \bra{j} \, (\cD_m \circ \mS(\vec{\cC}) \circ \cE_m) \lrp{\proj{j}} \, \ket{j}$.
Minimizing over all possible supermaps and over all possible encoding and decoding channels, we then obtain the minimum error probability
\begin{align}\label{eq:min_error_proba_def}
\begin{aligned}
    p^{\bf X}( \vec{\cC},m )&: = \min_{\mS\in {\bf X}_{AB}} \min_{\cE_m, \cD_m}\, p(\vec{\cC},\mS,m,\cE_m, \cD_m)\\
    &= \min_{\mS'\in {\bf X}_{A_mB_m}}\, 1 - \frac{1}{m}\sum_{j=0}^{m-1} \bra{j} \, \mS'(\,\vec{\cC}\,) \lrp{\proj{j}} \, \ket{j},
\end{aligned}
\end{align}
where the second equality follows from the fact that $\bf X$ is closed under local channels on the sender's and receiver's end.

In Supplementary Note~\ref{supp_sec:capacity}, we show that the minimum error probability is equal to the minimum diamond norm
\begin{align}\label{eq:min_diamond_norm_def}
    \omega^{\bf X}(\vec{\cC}, \Delta_m) \coloneqq \min_{\mS \in {\bf X}_{A_mB_m}} \|\mS(\,\vec{\cC}\,) - \Delta_m\|_\diamond / 2,
\end{align}
where $\Delta_m$ is the ideal classical channel for the transmission of $m$ messages.
Hence, the maximum number of classical bits that can be reliably transmitted (with average error probability bounded by $\epsilon$) through a single use of the channels $\vec{\cC}$ assisted by supermaps in $\bf X$ is
\begin{align}
\begin{aligned}\label{eq:capacity_def}
    C_\epsilon^{\bf X} (\,\vec{\cC}\,) \coloneqq \log_2 \max &\; m\\
    \text{\rm s.t.} &\; \omega^{\bf X}(\vec{\cC}, \Delta_m) \le \epsilon.
\end{aligned}
\end{align}
We call $C_\epsilon^{\bf X} (\,\vec{\cC}\,)$ the one-shot classical capacity of the channels $\vec{\cC}$ assisted by supermaps in ${\bf X}$.

In Supplementary Note~\ref{supp_sec:capacity}, we further show that the evaluation of $C_\epsilon^{\bf X} (\,\vec{\cC}\,)$ can be reduced to an SDP.
In the particular case where $\bf X$ is the set of supermaps with parallel placement, the SDP was provided in Refs.~\cite{leung2015power,wang2017semidefinite}.
When $\bf X$ is the set of all free supermaps in our resource theory, we show that the SDP is
\begin{align}
\begin{aligned}
    C_\epsilon^{\bf Free} (\,\vec{\cC}\,) = \log_2 \max &\; \left\lfloor m\right\rfloor\\
    \text{\rm s.t.} &\; E_{XY} \star J^{\vec{\cC}}_{XY} \ge m(1 - \epsilon),\\
    &\; E_{Y} = \idop_{Y},\, 0 \leq E_{XY} \leq F_{XY},\\
    &\; \cL^{\rm NS}_{XY}\lrp{F_{XY}} = \frac{m}{d_{X}}\idop_{XY},
\end{aligned}
\end{align}
where $m$ is a real number, $E_{XY}$ and $F_{XY}$ are Hermitian matrices, and $\star$ denotes the link product operation~\cite{chiribella2008quantum,chiribella2009theoretical}, defined as $P_{AB} \star Q_{BC} \coloneqq \tr_B\lrb{\lrp{P_{AB} \ox \idop_C} \lrp{\idop_A \ox Q_{BC}^{{\rm T}_B}}}$, where ${\rm T}_B$ is the partial transpose over system $B$, and $P_{AB}$ and $Q_{BC}$ are arbitrary linear operators on the Hilbert space of systems $AB$ and $BC$, respectively.

Similarly, the one-shot classical capacity assisted by supermaps in ${\bf FreeDef}$ can also be formulated as an SDP, as shown in Supplementary Note~\ref{supp_sec:capacity}.
These SDP formulations allow us to rigorously prove the advantage of indefinite causal order (see Supplementary Note~\ref{supp_sec:advantage} for the proof) as well as to numerically compute the capacities for different $\vec{\cC}$.

%%%%%%%%%%%%%%%%%%%%%%%%%%%%%%%%%%%%%%%%%%%%%%%%%%%%%%%%%%%%
%%%%%%%%%%%%%%%%%%%%%%%%%%%%%%%%%%%%%%%%%%%%%%%%%%%%%%%%%%%%
\section*{Acknowledgments}
We thank Jessica Bavaresco, Timoth{\'e}e Hoffreumon, Marco T{\'u}lio Quintino, Bartosz Regu{\l}a, Xin Wang, and Chengkai Zhu for helpful discussions.
This work is supported by the Chinese Ministry of Science and Technology (MOST) through grant 2023ZD0300600, by the Hong Kong Research Grant Council (RGC) through grants SRFS2021-7S02 and R7035-21F, by the State Key Laboratory of Quantum Information Technologies and Materials, Chinese University of Hong Kong, and by the Agence Nationale de la Recherche (ANR) through the project ANR-22-CE47-0012.
Research at the Perimeter Institute is supported by the Government of Canada through the Department of Innovation, Science and Economic Development Canada and by the Province of Ontario through the Ministry of Research, Innovation and Science.
For the purpose of open access, the authors have applied a CC-BY license to any author-accepted manuscript version arising from this submission.

%%%%%%%%%%%%%%%%%%%%%%%%%%%%%%%%%%%%%%%%%%%%%%%%%%%%%%%%%%%%
% Bibliography
%%%%%%%%%%%%%%%%%%%%%%%%%%%%%%%%%%%%%%%%%%%%%%%%%%%%%%%%%%%%
\bibliography{ref}

%%%%%%%%%%%%%%%%%%%%%%%%%%%%%%%%%%%%%%%%%%%%%%%%%%%%%%%%%%%%
%%%%%%%%%%%%%%%%%%%%%%%%%%%%%%%%%%%%%%%%%%%%%%%%%%%%%%%%%%%%
\clearpage
\vspace{2cm}
\onecolumngrid
\vspace{2cm}
\begin{center}
{\textbf{\large Supplementary Material for \\ The communication power of indefinite causal order}}
\end{center}
\appendix

%%%%%%%%%%%%%%%%%%%%%%%%%%%%%%%%%%%%%%%%%%%%%%%%%%%%%%%%%%%%
\renewcommand{\appendixname}{Supplementary Note}
\renewcommand{\thesection}{\arabic{section}}
\renewcommand{\thesubsection}{\alph{subsection}}
\renewcommand{\theequation}{\arabic{section}.\arabic{equation}}
\renewcommand{\theHequation}{\arabic{section}.\arabic{equation}}

%%%%%%%%%%%%%%%%%%%%%%%%%%%%%%%%%%%%%%%%%%%%%%%%%%%%%%%%%%%%%%%%%%%%%%%%%%%
%%%%%%%%%%%%%%%%%%%%%%%%%%%%%%%%%%%%%%%%%%%%%%%%%%%%%%%%%%%%%%%%%%%%%%%%%%%
\section{Characterizations of supermaps with different causal structures}\label{supp_sec:causal_order}
In Methods, we mentioned that each supermap mapping $N$ channels with input/output pairs $(X_i, Y_i)_{i=1}^N$ into a quantum channel with input $A$ and output $B$ can be represented by a unique linear operator $J_{XYAB}$, known as its Choi operator.
In this note, we specify the characterizations of supermaps with different causal structures in terms of linear constraints on their Choi operators.

%%%%%%%%%%%%%%%%%%%%%%%%%%%%%%%%%%%%%%%%%%%%%%%%%%%%%%%%%%%%
\subsection{Causally fixed supermaps}
Causally fixed supermaps, also known as quantum combs, are transformations between quantum channels that can be implemented by inserting input quantum channels into quantum circuits with open slots in a fixed order.
These supermaps consist of both parallel and sequential strategies, where sequential strategies are sometimes also known as adaptive strategies.
Such supermaps are introduced in Refs.~\cite{chiribella2008transforming,chiribella2008quantum,chiribella2009theoretical} with the following characterization.

\begin{proposition}[Characterization of causally fixed supermaps~\cite{chiribella2008quantum}]\label{prop:fix_char}
    A linear operator $J_{XYAB}$ is the Choi operator of some causally fixed supermap mapping $N$ channels with input/output pairs $(X_i,Y_i)_{i=1}^N$ into a channel with input $A$ and output $B$ if and only if
    \begin{align}
        J_{XYAB} \geq 0,\quad {}_{\lrb{1-Y_n}}J_{X_1Y_1\dots X_nY_n A} = 0 \quad\forall\; n = 1,\dots,N,\quad J_{A} = d_{Y}\idop_{A}.
    \end{align}
\end{proposition}

Parallel strategies correspond to the simplest type of quantum combs, the single-slot ones:
\begin{align}
    J_{XYAB} \geq 0,\quad {}_{\lrb{1-Y}}J_{XYA} = 0,\quad J_{A} = d_{Y}\idop_{A}.
\end{align}
In the above characterization and the rest of the paper, we use the notation ${}_Y J \coloneqq \tr_Y\lrb{J} \ox \idop_Y/d_Y$, ${}_1 J \coloneqq J$, and ${}_{\lrb{\sum_j s_j Y_j}}J \coloneqq \sum_j s_j \cdot {}_{Y_j} J$ for $s_j\in\{-1,1\}$.
For example, ${}_{\lrb{1-Y}}J_{XYA} = {}_1 J_{XYA} - {}_Y J_{XYA} = J_{XYA} - J_{XA} \ox I_Y / d_Y$.
In these equations, $d_Y$ denotes the dimension of system $Y$ (and we use a similar notation for any other system). Furthermore, recall that we use the convention that dropping a subscript on an operator means tracing the corresponding system out: {\em e.g.}, $J_{XYA}=\tr_B [J_{XYAB}]$, $J_A=\tr_{XYB} [J_{XYAB}]$.

%%%%%%%%%%%%%%%%%%%%%%%%%%%%%%%%%%%%%%%%%%%%%%%%%%%%%%%%%%%%
\subsection{QC-CC operations and causally definite supermaps}
Causally fixed supermaps are not the most general type of supermaps allowed by definite causal order.
For example, by dynamically controlling the order of input channels in a classical manner, one can realize supermaps that are causally non-fixed but still causally definite~\cite{oreshkov2016causal,wechs2021quantum}.
We refer to such supermaps realizable by quantum circuits with classical control (QC-CC) as QC-CC operations.

\begin{proposition}[Characterization of QC-CC operations~\cite{wechs2021quantum}]\label{prop:def_char}
    A linear operator $J_{XYAB}$ is the Choi operator of some QC-CC operation mapping $N$ channels with input/output pairs $(X_i,Y_i)_{i=1}^N$ into a channel with input $A$ and output $B$ if and only if there exist Hermitian operators $J^{(k_1,\dots,k_N)}_{XYAB}$ for all $(k_1,\dots,k_N)$, where $(k_1,\dots,k_N)$ enumerates all ordered sequences consisting of numbers $1$ to $N$ with each number appearing once in each sequence, such that
    \begin{align}
    \begin{gathered}
        J_{XYAB} = \sum_{(k_1,\dots,k_N)} J^{(k_1,\dots,k_N)}_{XYAB},\quad \sum_{(k_1,\dots,k_N)} J^{(k_1,\dots,k_N)}_{A} = d_{Y}\idop_{A},\\
        J^{(k_1,\dots,k_N)}_{XYAB} \geq 0 \quad\forall\; (k_1,\dots,k_N),\\
        {}_{\lrb{1-Y_{k_n}}}\lrp{\sum_{(k_{n+1},\dots,k_N)} J^{(k_1,\dots,k_n,k_{n+1},\dots,k_N)}_{X_{k_1}Y_{k_1}\dots X_{k_n}Y_{k_n} A}} = 0 \quad\forall\; n=1,\dots,N-1 \quad\text{and}\quad \forall\; (k_1,\dots,k_n),\\
        {}_{\lrb{1-Y_{k_N}}}J^{(k_1,\dots,k_N)}_{XYA} = 0 \quad \forall\; (k_1,\dots,k_N),
    \end{gathered}
    \end{align}
    where the sum in the third line is over all sequences $(k_{n+1},\dots,k_N)$ that complete $(k_1,\dots,k_n)$, i.e., such that $(k_1,\dots,k_n,k_{n+1},\dots,k_N)$ contains all numbers from 1 to $N$.
\end{proposition}

It has been shown that QC-CC operations are equivalent to causally definite supermaps when the number of input channels is $N=2$ (or $N=3$ if $A$ is trivial, {\em i.e.}, $d_A=1$)~\cite{wechs2019on,wechs2021quantum}.
Whether this equivalence holds more generally for $N \geq 3$ still remains an open problem.
Since our analysis of the communication advantage of indefinite causal order is limited to the case $N = 2$, it suffices to consider only QC-CC operations.

%%%%%%%%%%%%%%%%%%%%%%%%%%%%%%%%%%%%%%%%%%%%%%%%%%%%%%%%%%%%
\subsection{QC-QC operations}
Beyond classical control of the order as in QC-CC operations, one can consider supermaps that apply operations in a quantum, coherently controlled order~\cite{wechs2021quantum,purves2021quantum}.
Such quantum circuits with quantum control of causal order (QC-QC) can exhibit indefinite causal order and can be reproduced in photonic experiments.
They are characterized as follows:

\begin{proposition}[Characterization of QC-QC operations~\cite{wechs2021quantum}]\label{prop:qcqc_char}
    A linear operator $J_{XYAB}$ is the Choi operator of some QC-QC operation mapping $N$ channels with input/output pairs $(X_i,Y_i)_{i=1}^N$ into a channel with input $A$ and output $B$ if and only if there exist Hermitian operators $J^{(\hat{K}_n,k_{n+1})}_{X_{\hat{K}_n}X_{k_{n+1}}Y_{\hat{K}_n}A}$ for all strict subsets $\hat{K}_n$ of $\hat{N} \coloneqq \{1,\dots,N\}$ with $|\hat{K}_{n}| = n$ and $k_{n+1} \in \hat{N}\backslash\hat{K}_n$ such that
    \begin{align}
    \begin{gathered}
        J_{XYAB} \geq 0,\quad J^{(\hat{K}_n,k_{n+1})}_{X_{\hat{K}_n}X_{k_{n+1}}Y_{\hat{K}_n}A} \geq 0 \quad\forall\; \hat{K}_n \subsetneq \hat{N} \,\text{and}\, k_{n+1} \in \hat{N}\backslash\hat{K}_n,\\
        J_{XYA} = \sum_{k_N \in \hat{N}} J^{(\hat{N}\backslash\{k_N\}, k_N)}_{XY_{\hat{N}\backslash\{k_N\}}A} \ox \frac{\idop_{Y_{k_N}}}{d_{Y_{k_N}}},\\
        \sum_{k_{n+1} \in \hat{N}\backslash\hat{K}_n} J^{(\hat{K}_n,k_{n+1})}_{X_{\hat{K}_n}Y_{\hat{K}_n}A} = \sum_{k_n \in \hat{K}_n} J^{(\hat{K}_n\backslash\{k_n\},k_n)}_{X_{\hat{K}_n}Y_{\hat{K}_n\backslash\{k_n\}}A} \ox \frac{\idop_{Y_{k_n}}}{d_{Y_{k_n}}} \quad\forall\; \emptyset \subsetneq \hat{K}_n \subsetneq \hat{N},\\
        \sum_{k_1 \in \hat{N}} J^{(\emptyset,k_1)}_A = d_Y \idop_A,
    \end{gathered}
    \end{align}
    where $X_{\hat{K}_n}$ ($Y_{\hat{K}_n}$) is a shorthand for $X_{k_1}\cdots X_{k_n}$ ($Y_{k_1}\cdots Y_{k_n}$) with $\hat{K}_n = \{k_1,\dots,k_n\}$.
\end{proposition}

%%%%%%%%%%%%%%%%%%%%%%%%%%%%%%%%%%%%%%%%%%%%%%%%%%%%%%%%%%%%
\subsection{General supermaps}
The most general causal structure allows for indefinite causal order even beyond QC-QC operations~\cite{chiribella2013quantum,oreshkov2012quantum}.
The characterization of general, possibly causally indefinite supermaps was briefly discussed in Methods.
The formalism of process matrices was first proposed in Ref.~\cite{oreshkov2012quantum} to describe all possible transformations, including causally indefinite ones, from multipartite local quantum operations to probability distributions.
Such transformations are named as processes.
In this work, a general supermap mapping $N$ channels to a new channel can be regarded as a process involving $N+2$ parties.
The first $N$ parties have input/output pairs $(X_i,Y_i)_{i=1}^N$.
The $(N+1)$-th party has an output system $A$ and a trivial input system.
The $(N+2)$-th party has an input system $B$ and a trivial output system.
According to Ref.~\cite{araujo2015witnessing}, a linear operator $J_{XYAB}$ is the Choi operator of such a general supermap if and only if
\begin{align}
\begin{gathered}\label{supp_eq:gen_constraints_init}
    J_{XYAB} \geq 0,\quad \tr\lrb{J_{XYAB}} = d_{YA},\quad \cL\lrp{J_{XYAB}} = J_{XYAB},
\end{gathered}
\end{align}
where $\cL$ is a projection onto a linear subspace, defined by $\cL\lrp{J_{XYAB}} \coloneqq {}_{\lrb{1 - \lrp{\prod_{j=1}^{N}\lrp{1-Y_j+X_jY_j}}B + XYAB}} J_{XYAB}$.
We now show that this above characterization is equivalent to a simpler one that has a clear physical interpretation.

\begin{proposition}[Characterization of general supermaps]\label{prop:gen_char}
    A linear operator $J_{XYAB}$ is the Choi operator of some general supermap mapping $N$ channels with input/output pairs $(X_i,Y_i)_{i=1}^N$ into a channel with input $A$ and output $B$ if and only if
    \begin{align}
        J_{XYAB} \geq 0 \quad\text{and}\quad \cL^{\rm NS}_{XY}\lrp{J_{XYA}} = \frac{\idop_{XYA}}{d_{X}},
    \end{align}
    where $\cL^{\rm NS}_{XY}\lrp{J_{XYA}}$ is a shorthand for $(\cL^{\rm NS}_{XY} \ox \cI_A) \lrp{J_{XYA}}$, $\cI_A$ is the identity map acting on system $A$, and $\cL^{\rm NS}_{XY}$ is a linear map acting on systems $XY$ as $\cL^{\rm NS}_{XY}\lrp{J_{XY}} \coloneqq {}_{\lrb{\prod_{j=1}^{N}\lrp{1-Y_j+X_jY_j}}} J_{XY}$.
\end{proposition}
\begin{proof}
    Consider that
    \begin{align}
        \cL\lrp{J_{XYAB}} &= J_{XYAB} - {}_{\lrb{\prod_{j=1}^{N}\lrp{1-Y_j+X_jY_j}}} J_{XYA} \ox \frac{\idop_{B}}{d_{B}} + \frac{\tr\lrb{J_{XYAB}}\idop_{XYAB}}{d_{XYAB}}\\
        &= J_{XYAB} - \lrp{{}_{\lrb{\prod_{j=1}^{N}\lrp{1-Y_j+X_jY_j}}} J_{XYA} - \frac{\idop_{XYA}}{d_{X}}} \ox \frac{\idop_{B}}{d_{B}},
    \end{align}
    where the second equality follows from $\tr\lrb{J_{XYAB}} = d_{YA}$ as given in Eq.~\eqref{supp_eq:gen_constraints_init}.
    Then, to require $\cL\lrp{J_{XYAB}} = J_{XYAB}$ is equivalent to requiring
    \begin{align}
        \cL^{\rm NS}_{XY}\lrp{J_{XYA}} = {}_{\lrb{\prod_{j=1}^{N}\lrp{1-Y_j+X_jY_j}}} J_{XYA} = \frac{\idop_{XYA}}{d_{X}}.
    \end{align}
    Note that the projector $\cL^{\rm NS}_{XY}$ is trace-preserving. Therefore, the above constraint implies
    \begin{align}
        \tr\lrb{J_{XYAB}} = \tr\lrb{\cL^{\rm NS}_{XY}\lrp{J_{XYA}}} = d_{YA},
    \end{align}
    which completes the proof.
\end{proof}

This characterization also applies to general processes involving $N$ local parties as long as we set both $A$ and $B$ to trivial one-dimensional systems.
Note that we use the superscript ${\rm NS}$ in $\cL^{\rm NS}_{XY}$ to emphasize the fact that $\cL^{\rm NS}_{XY}$ is actually the orthogonal projection onto the subspace of Hermitian operators spanned by the Choi operators of no-signaling multipartite quantum channels, as we show below.

\begin{proposition}[Projection onto no-signaling subspace]
    $\cL^{\rm NS}_{XY}$ is the orthogonal projection of Hermitian operators onto the subspace spanned by the Choi operators of no-signaling channels from $X=X_1\cdots X_N$ to $Y=Y_1\cdots Y_N$.
\end{proposition}
\begin{proof}
    We call an $N$-partite channel (with input/output pairs $(X_i,Y_i)$ for each of the $N$ parties) no-signaling if the input state at any one local input system does not affect the output state at the rest $N-1$ local output systems.
    In terms of its Choi operator $J_{XY}$, this is equivalent to require
    \begin{align}
        \tr_{Y_j}\lrb{J_{XY}} = \tr_{X_jY_j}\lrb{J_{XY}} \ox \frac{\idop_{X_j}}{d_{X_j}} \quad\forall\; j = 1,\dots,N.
    \end{align}
    These no-signaling constraints can be equivalently written as
    \begin{align}
        {}_{\lrb{Y_j - X_jY_j}}J_{XY} = 0 \quad\forall\; j = 1,\dots,N.
    \end{align}
    Note that the no-signaling constraints imply ${}_{Y}J_{XY} = \tr\lrb{J_{XY}}\idop_{XY} / d_{XY}$, which says that $J_{XY}$ is proportional to the Choi operator of some trace-preserving map.
    Due to the linearity of this set of constraints, the subspace spanned by Choi operators of no-signaling channels consists of all Hermitian operators satisfying the same set of constraints.
    For each $j$, the constraint ${}_{\lrb{Y_j - X_jY_j}}J_{XY} = 0$ is satisfied if and only if
    \begin{align}
        \cL^{{\rm NS}}_{X_jY_j}\lrp{J_{XY}} \coloneqq {}_{\lrb{1 - Y_j + X_jY_j}}J_{XY} = J_{XY}.
    \end{align}
    Furthermore, by checking that $\cL^{{\rm NS}}_{X_jY_j} \circ \cL^{{\rm NS}}_{X_jY_j} = \cL^{{\rm NS}}_{X_jY_j}$ and it is self-adjoint, one can verify that $\cL^{{\rm NS}}_{X_jY_j}$ is the orthogonal projection onto the subspace of Hermitian operators satisfying ${}_{\lrb{Y_j - X_jY_j}}J_{XY} = 0$. The intersection of subspaces associated with all $\cL^{{\rm NS}}_{X_jY_j}$ is exactly the no-signaling subspace. Since all the projections $\cL^{{\rm NS}}_{X_jY_j}$ commute as they all act on different systems, the orthogonal projection onto the no-signaling subspace is given by
    \begin{align}
        \cL^{\rm NS}_{X_1Y_1} \circ \cdots \circ \cL^{\rm NS}_{X_NY_N} = \cL^{\rm NS}_{XY},
    \end{align}
    which completes the proof.
\end{proof}

%%%%%%%%%%%%%%%%%%%%%%%%%%%%%%%%%%%%%%%%%%%%%%%%%%%%%%%%%%%%
\subsection{Characterization of signaling-non-generating supermaps}
Our resource theory considers as free those supermaps that transform lists of replacement channels into replacement channels. These signaling-non-generating supermaps can be characterized as follows:

\begin{proposition}
    A general supermap transforming $N$ channels with input/output pairs $(X_i,Y_i)_{i=1}^N$ into a channel with input $A$ and output $B$ is signaling-non-generating if and only if it satisfies the No Forward Signaling Condition: its corresponding quantum channel, with input $YA$ and output $XB$, is no-signaling from $A$ to $B$.
\end{proposition}

\begin{proof}
    Let $J_{XYAB}$ be the Choi operator of a given general supermap.
    Because the Choi operator of a replacement channel is of the form $\idop_{X_i} \ox \rho^{(i)}_{Y_i}$ with $\rho^{(i)}_{Y_i}$ being a density matrix, the supermap being signaling-non-generating is equivalent to $J_{XYAB} \star (\idop_X \ox \rho^{(1,\dots,N)}_Y) = J_{YAB} \star \rho^{(1,\dots,N)}_Y$ being the Choi operator of a replacement channel for every possible $\rho^{(1,\dots,N)}_Y$, where $\rho^{(1,\dots,N)}_Y \coloneqq \bigotimes_{i=1}^N \rho^{(i)}_{Y_i}$.
    The latter condition can be written as ${}_{[1-A]} (J_{YAB} \star \rho^{(1,\dots,N)}_Y) = ({}_{[1-A]} J_{YAB}) \star \rho^{(1,\dots,N)}_Y = 0$ for all $\rho^{(1,\dots,N)}_Y$, and this is true if and only if
    \begin{align}\label{supp_eq:forward_no_sig_init}
        {}_{[1-A]} J_{YAB} = 0,
    \end{align}
    {\em i.e.}, if and only if the channel from $YA$ to $XB$ described by $J_{XYAB}$ is no-signaling from $A$ to $B$.
\end{proof}

%%%%%%%%%%%%%%%%%%%%%%%%%%%%%%%%%%%%%%%%%%%%%%%%%%%%%%%%%%%%%%%%%%%%%%%%%%%
%%%%%%%%%%%%%%%%%%%%%%%%%%%%%%%%%%%%%%%%%%%%%%%%%%%%%%%%%%%%%%%%%%%%%%%%%%%
\section{One-shot classical communication capacity}\label{supp_sec:capacity}
%%%%%%%%%%%%%%%%%%%%%%%%%%%%%%%%%%%%%%%%%%%%%%%%%%%%%%%%%%%%%%%%%%%%%%%%%%%
\subsection{Equivalence between minimum error and minimum average probability of error}
\begin{proposition}\label{supp_prop:diamond_prob_equal}
    Let ${\bf X}_{A_mB_m}$ be a convex set of general supermaps mapping $N$ channels with input/output pairs $(X_i,Y_i)_{i=1}^N$ into a channel with input $A_m$ and output $B_m$, where both $A_m$ and $B_m$ are $m$-dimensional systems.
    Given a list $\vec{\cC}$ of $N$ channels, it holds that
    \begin{align}\label{appeq:diamond_prob_equal}
        \omega^{\bf X}(\vec{\cC}, \Delta_m) = p^{\bf X}(\vec{\cC}, m)
    \end{align}
    if ${\bf X}_{A_mB_m}$ is closed under simultaneous permutations on $A_m$ and $B_m$, i.e., for every supermap $\mS \in {\bf X}_{A_mB_m}$ and for every permutation $\tau$ in the symmetric group $S_m$ of degree $m$, the supermap $\mW^\tau \circ \mS \in {\bf X}_{A_mB_m}$, where $\mW^\tau$ acts as $\mW^\tau(\cdot) \coloneqq \cW^\tau_{B\to B} \circ (\cdot) \circ \cW^{\tau^\dagger}_{A\to A}$ with $\cW^\tau$ being the permutation channel associated with $\tau$.
\end{proposition}
\begin{proof}
    Recall the definitions of $p^{\bf X}(\vec{\cC}, m)$ and $\omega^{\bf X}(\vec{\cC}, \Delta_m)$ from Eqs.~\eqref{eq:min_error_proba_def}--\eqref{eq:min_diamond_norm_def} in the Methods.

    For $m = 1$, it is clear that $\omega^{\bf X}(\vec{\cC}, \Delta_1) = p^{\bf X}(\vec{\cC}, 1) = 0$ because there is only one possible state for a one-dimensional quantum system, which makes the communication task trivial.
    Hence, Eq.~\eqref{appeq:diamond_prob_equal} does hold for $m = 1$.

    Now, we consider cases where $m \geq 2$. For a list of $N$ quantum channels $\vec{\cC} = (\cC_1,\dots,\cC_N)$, we can express the minimum average error probability $p^{\bf X}(\vec{\cC}, m)$ in terms of its Choi operator $J^{\vec{\cC}}_{XY} \coloneqq \bigotimes_{i=1}^N J^{\cC_i}_{X_iY_i}$ ($J^{\cC_i}_{X_iY_i}$ is the Choi operator of $\cC_i$) as
    \begin{align}
    &\begin{aligned}
        p^{\bf X}(\vec{\cC}, m) = \min &\; 1 - \frac{1}{m}\sum_{j=0}^{m-1} \tr\lrb{J^\cM_{A_mB_m} \lrp{\proj{j}_{A_m} \ox \proj{j}_{B_m}}}\\
        \text{\rm s.t.} &\; J^\cM_{A_mB_m} = J^\mS_{XYA_mB_m} \star J^{\vec{\cC}}_{XY},\, J^\mS_{XYA_mB_m} \in \widetilde{\bf X}_{A_mB_m}
    \end{aligned}\\
    &\phantom{p^{\bf X}(\vec{\cC}, m)}
    \begin{aligned}\label{appeq:min_err_prob_program}
        \mathllap{} = \min &\; 1 - \frac{1}{m} \tr\lrb{\lrp{J^\mS_{XYA_mB_m} \star J^{\vec{\cC}}_{XY}} J^{\Delta_m}_{A_mB_m}}\\
        \text{\rm s.t.} &\; J^\mS_{XYA_mB_m} \in \widetilde{\bf X}_{A_mB_m},
    \end{aligned}
    \end{align}
    where $J^{\Delta_m}_{A_mB_m} \coloneqq \sum_{j=0}^{m-1} \proj{j}_{A_m} \ox \proj{j}_{B_m}$ is the Choi operator of an $m$-dimensional classical noiseless channel $\Delta_m$, and $\widetilde{\bf X}_{A_mB_m}$ denotes the set of all Choi operators corresponding to supermaps in the given set ${\bf X}_{A_mB_m}$.

    The optimization program in Eq.~\eqref{appeq:min_err_prob_program} can be further simplified by exploiting the symmetry of $J^{\Delta_m}_{A_mB_m}$.
    Specifically, if $J^\mS_{XYA_mB_m}$ is optimal with the corresponding supermap $\mS \in {\bf X}_{A_mB_m}$, then the corresponding Choi operator $J^{\mS^\tau}_{XYA_mB_m}$ of $\mS^\tau \coloneqq \mW^\tau \circ \mS$ is feasible for any permutation $\tau$ in the symmetric group $S_m$ of degree $m$ because $\mS^\tau \in {\bf X}_{A_mB_m}$.
    Given that
    \begin{align}
        J^{\mW^\tau}_{A_mB_mA'_mB'_m} &= \lrp{\sum_{j,k = 0}^{m-1} \ketbra{j}{k}_{A'_m} \ox \ketbra{\tau^\dagger(j)}{\tau^\dagger(k)}_{A_m}} \ox \lrp{\sum_{j',k'= 0}^{m-1} \ketbra{j'}{k'}_{B_m} \ox \ketbra{\tau(j')}{\tau(k')}_{B'_m}}\\
        &= \lrp{\sum_{j,k = 0}^{m-1} \ketbra{\tau(j)}{\tau(k)}_{A'_m} \ox \ketbra{j}{k}_{A_m}} \ox \lrp{\sum_{j',k'= 0}^{m-1} \ketbra{j'}{k'}_{B_m} \ox \ketbra{\tau(j')}{\tau(k')}_{B'_m}}\\
        &= \sum_{j,k,j',k' = 0}^{m-1} \ketbra{j}{k}_{A_m} \ox \ketbra{j'}{k'}_{B_m} \ox \ketbra{\tau(j)}{\tau(k)}_{A'_m} \ox \ketbra{\tau(j')}{\tau(k')}_{B'_m},
    \end{align}
    where we mark the output system as $A'_mB'_m$ for clarity, the Choi operator $J^{\mS^\tau}_{XYA_mB_m}$ can be written as
    \begin{align}
        J^{\mS^\tau}_{XYA_mB_m} &= J^\mS_{XYA_mB_m} \star J^{\mW^\tau}_{A_mB_mA'_mB'_m}\\
        &= \tr_{A_mB_m}\lrb{\lrp{J^\mS_{XYA_mB_m} \ox \idop_{A'_mB'_m}} \lrp{\idop_{XY} \ox \sum_{j,k,j',k' = 0}^{m-1} \ketbra{k}{j}_{A_m} \ox \ketbra{k'}{j'}_{B_m} \ox \ketbra{\tau(j)}{\tau(k)}_{A'_m} \ox \ketbra{\tau(j')}{\tau(k')}_{B'_m}}}\\
        &= \sum_{j,k,j',k' = 0}^{m-1} \tr_{A_mB_m}\lrb{\lrp{J^\mS_{XYA_mB_m} \ox \idop_{A'_mB'_m}} \lrp{\idop_{XY} \ox \ketbra{k}{j}_{A_m} \ox \ketbra{k'}{j'}_{B_m} \ox \ketbra{\tau(j)}{\tau(k)}_{A'_m} \ox \ketbra{\tau(j')}{\tau(k')}_{B'_m}}}\\
        &= \sum_{j,k,j',k' = 0}^{m-1} \tr_{A_mB_m}\lrb{J^\mS_{XYA_mB_m} \lrp{\idop_{XY} \ox \ketbra{k}{j}_{A_m} \ox \ketbra{k'}{j'}_{B_m}}} \ox \ketbra{\tau(j)}{\tau(k)}_{A'_m} \ox \ketbra{\tau(j')}{\tau(k')}_{B'_m}\\
        &= \sum_{j,k,j',k' = 0}^{m-1} \lrp{\idop_{XY} \ox \bra{j}_{A_m} \ox \bra{j'}_{B_m}} J^\mS_{XYA_mB_m} \lrp{\idop_{XY} \ox \ket{k}_{A_m} \ox \ket{k'}_{B_m}} \ox \ketbra{\tau(j)}{\tau(k)}_{A'_m} \ox \ketbra{\tau(j')}{\tau(k')}_{B'_m}\\
        &= \sum_{j,k,j',k' = 0}^{m-1} \lrp{\idop_{XY} \ox \ket{\tau(j)}_{A'_m}\bra{j}_{A_m} \ox \ket{\tau(j')}_{B'_m}\bra{j'}_{B_m}} J^\mS_{XYA_mB_m} \lrp{\idop_{XY} \ox \ket{k}_{A_m}\bra{\tau(k)}_{A'_m} \ox \ket{k'}_{B_m}\bra{\tau(k')}_{B'_m}}\\
        &= \lrp{\idop_{XY} \ox \sum_{j = 0}^{m-1} \ket{\tau(j)}_{A'_m}\bra{j}_{A_m} \ox \sum_{j' = 0}^{m-1} \ket{\tau(j')}_{B'_m}\bra{j'}_{B_m}} J^\mS_{XYA_mB_m}\nonumber\\
        &\qquad\qquad \lrp{\idop_{XY} \ox \sum_{k = 0}^{m-1} \ket{k}_{A_m}\bra{\tau(k)}_{A'_m} \ox \sum_{k' = 0}^{m-1} \ket{k'}_{B_m}\bra{\tau(k')}_{B'_m}}\\
        &= \lrp{\idop_{XY} \ox \tau_{A_m} \ox \tau_{B_m}} J^\mS_{XYA_mB_m} \lrp{\idop_{XY} \ox \tau_{A_m} \ox \tau_{B_m}}^\dagger.
    \end{align}
    More importantly, $J^{\mS^\tau}_{XYA_mB_m}$ achieves the same average probability of error as $J^\mS_{XYA_mB_m}$ since
    \begin{align}
        \tr\lrb{\lrp{J^{\mS^\tau}_{XYA_mB_m} \star J^{\vec{\cC}}_{XY}} J^{\Delta_m}_{A_mB_m}} &= \tr\lrb{\lrp{\lrp{\idop_{XY} \ox \tau_{A_m} \ox \tau_{B_m}} J^\mS_{XYA_mB_m} \lrp{\idop_{XY} \ox \tau_{A_m} \ox \tau_{B_m}}^\dagger \star J^{\vec{\cC}}_{XY}} J^{\Delta_m}_{A_mB_m}}\\
        &= \tr\lrb{\lrp{J^\mS_{XYA_mB_m} \star J^{\vec{\cC}}_{XY}} \lrp{\tau_{A_m}\ox \tau_{B_m}}^\dagger J^{\Delta_m}_{A_mB_m} \lrp{\tau_{A_m}\ox \tau_{B_m}}}\\
        &= \tr\lrb{\lrp{J^\mS_{XYA_mB_m} \star J^{\vec{\cC}}_{XY}} J^{\Delta_m}_{A_mB_m}},
    \end{align}
    where the last equality holds due to $\lrp{\tau_{A_m}\ox \tau_{B_m}}^\dagger J^{\Delta_m}_{A_mB_m} \lrp{\tau_{A_m}\ox \tau_{B_m}} = J^{\Delta_m}_{A_mB_m}$.
    Thus, the optimality of $J^\mS_{XYA_mB_m}$ implies the optimality of $J^{\mS^\tau}_{XYA_mB_m}$.
    Moreover, any convex combination of optimal Choi operators is still optimal due to the convexity of ${\bf X}_{A_mB_m}$ and the linearity of the average error probability, implying that
    \begin{align}
        J^{\mS'}_{XYA_mB_m} \coloneqq \frac{1}{m!}\sum_{\tau\in S_m} J^{\mS^\tau}_{XYA_mB_m}
    \end{align}
    is also optimal.
    Then, by Schur's Lemma, we can restrict the optimization range of the operator $J^\mS_{XYA_mB_m}$ to operators of the form
    \begin{align}
        J^\mS_{XYA_mB_m} = E'_{XY}\ox J^{\Delta_m}_{A_mB_m} + F'_{XY}\ox\lrp{\idop_{A_mB_m}-J^{\Delta_m}_{A_mB_m}}, \label{eq:gen_form_Schur}
    \end{align}
    where $E'_{XY}$ and $F'_{XY}$ are arbitrary Hermitian operators.
    Indeed, observe that $J^{\mS'}_{XYA_mB_m}$ commutes with $\idop_{XY} \ox \tau_{A_m} \ox \tau_{B_m}$ for every $\tau \in S_m$, where $\tau_{A_m} \ox \tau_{B_m}$ is a representation of $S_m$.
    This representation can be decomposed into two irreducible representations on the invariant subspace spanned by $\{\ket{j}_{A_m}\ket{j}_{B_m}\}_{j=0}^{m-1}$ and its orthogonal complement.
    Then, Schur's Lemma implies Eq.~\eqref{eq:gen_form_Schur}, where $J^{\Delta_m}_{A_mB_m}$ and $\lrp{\idop_{A_mB_m}-J^{\Delta_m}_{A_mB_m}}$ are the projectors onto the aforementioned invariant subspaces.

    With this, $J^\mS_{XYA_mB_m} \star J^{\vec{\cC}}_{XY}$ becomes
    \begin{align}
        \lrp{E'_{XY}\ox J^{\Delta_m}_{A_mB_m} + F'_{XY}\ox\lrp{\idop_{A_mB_m}-J^{\Delta_m}_{A_mB_m}}} \star J^{\vec{\cC}}_{XY} = \lrp{E'_{XY} \star J^{\vec{\cC}}_{XY}}J^{\Delta_m}_{A_mB_m} + \lrp{F'_{XY} \star J^{\vec{\cC}}_{XY}}\lrp{\idop_{A_mB_m}-J^{\Delta_m}_{A_mB_m}},
    \end{align}
    and noting that $\tr\lrb{J^{\Delta_m}_{A_mB_m} J^{\Delta_m}_{A_mB_m}} = m$ and $\tr\lrb{J^{\Delta_m}_{A_mB_m} \lrp{\idop_{A_mB_m}-J^{\Delta_m}_{A_mB_m}}} = 0$, the optimization program of Eq.~\eqref{appeq:min_err_prob_program} can be written as
    \begin{align}
    \begin{aligned}\label{supp_eq:optim_pb_pX}
        p^{\bf X}(\vec{\cC}, m) = \min &\; 1 - E'_{XY} \star J^{\vec{\cC}}_{XY}\\
        \text{\rm s.t.} &\; E'_{XY}\ox J^{\Delta_m}_{A_mB_m} + F'_{XY}\ox\lrp{\idop_{A_mB_m}-J^{\Delta_m}_{A_mB_m}} \in \widetilde{\bf X}_{A_mB_m}.
    \end{aligned}
    \end{align}

    For the minimum error $\omega^{\bf X}(\vec{\cC}, \Delta_m)$, we can use the semidefinite program (SDP)~\cite{watrous2009semidefinite} for the diamond distance between quantum channels to write it as the following optimization program:
    \begin{align}
    \begin{aligned}
        \omega^{\bf X}(\vec{\cC}, \Delta_m) = \min &\; \epsilon\\
        \text{\rm s.t.} &\; Z_{A_mB_m} \geq J^\mS_{XYA_mB_m} \star J^{\vec{\cC}}_{XY} - J^{\Delta_m}_{A_mB_m},\, Z_{A_mB_m} \geq 0,\, Z_{A_m} \leq \epsilon\idop_{A_m},\, J^\mS_{XYA_mB_m} \in \widetilde{\bf X}_{A_mB_m}.
    \end{aligned}
    \end{align}
    From this explicit optimization program, we can see that restricting the optimization range of the operator $J^\mS_{XYA_mB_m}$ to operators of the form $E'_{XY}\ox J^{\Delta_m}_{A_mB_m} + F'_{XY}\ox\lrp{\idop_{A_mB_m}-J^{\Delta_m}_{A_mB_m}}$, again, does not change the optimal value.
    Hence, we arrive at
    \begin{align}
    \begin{aligned}
        \omega^{\bf X}(\vec{\cC}, \Delta_m) = \min &\; \frac{1}{2}\lrV{\cM - \Delta_m}_\diamond\\
        \text{\rm s.t.} &\; E'_{XY}\ox J^{\Delta_m}_{A_mB_m} + F'_{XY}\ox\lrp{\idop_{A_mB_m}-J^{\Delta_m}_{A_mB_m}} \in \widetilde{\bf X},
    \end{aligned}
    \end{align}
    where $\cM$ denotes the quantum channel whose Choi operator is
    \begin{align}
        J^\cM_{A_mB_m} \coloneqq \lrp{E'_{XY} \star J^{\vec{\cC}}_{XY}} J^{\Delta_m}_{A_mB_m} + \lrp{F'_{XY} \star J^{\vec{\cC}}_{XY}} \lrp{\idop_{A_mB_m}-J^{\Delta_m}_{A_mB_m}}.
    \end{align}
    One constraint that holds for the Choi operator $J^{\mS}_{XYA_mB_m}$ of every general supermap $\mS$ is $\cL^{\rm NS}_{XY}\lrp{J^\mS_{XYA_m}} = \frac{\idop_{XYA_m}}{d_{X}}$ (see Proposition~\ref{prop:gen_char}).
    In terms of operators $E'_{XY}$ and $F'_{XY}$, this constraint can be expressed as
    \begin{align}\label{supp_eq:LNS_E_m1F}
        \cL^{\rm NS}_{XY}\lrp{E'_{XY} + (m-1)F'_{XY}} = \frac{\idop_{XY}}{d_{X}},
    \end{align}
    which implies that
    \begin{align}
        F'_{XY} \star J^{\vec{\cC}}_{XY} &= F'_{XY} \star \cL^{\rm NS}_{XY}\lrp{J^{\vec{\cC}}_{XY}}\\
        &= \cL^{\rm NS}_{XY}\lrp{F'_{XY}} \star J^{\vec{\cC}}_{XY}\\
        &= \frac{1}{m-1}\lrp{\frac{\idop_{XY}}{d_{X}} - \cL^{\rm NS}_{XY}\lrp{E'_{XY}}} \star J^{\vec{\cC}}_{XY}\\
        &= \frac{1}{m-1}\lrp{1 - E'_{XY} \star J^{\vec{\cC}}_{XY}},
    \end{align}
    where the first equality holds because $\vec{\cC}$ is no-signaling, the second equality is due to that $\cL^{\rm NS}_{XY}$ is self-dual, and the last equality uses these two properties together with the fact that $\idop_{XY} \star J^{\vec{\cC}}_{XY} = \tr\lrb{J^{\vec{\cC}}_{XY}} = d_X$.
    It then follows that
    \begin{align}
        J^\cM_{A_mB_m} - J^{\Delta_m}_{A_mB_m} &= \lrp{E'_{XY} \star J^{\vec{\cC}}_{XY} - 1} J^{\Delta_m}_{A_mB_m} + \frac{1}{m-1}\lrp{1 - E'_{XY} \star J^{\vec{\cC}}_{XY}} \lrp{\idop_{A_mB_m}-J^{\Delta_m}_{A_mB_m}}\\
        & = \frac{1 - E'_{XY} \star J^{\vec{\cC}}_{XY}}{m - 1} \lrp{\idop_{A_mB_m} - m J^{\Delta_m}_{A_mB_m}}.
    \end{align}

    It is known that the diamond norm can be written as the following optimization~\cite{wilde2013quantum}:
    \begin{align}
        \lrV{\cM - \Delta_m}_\diamond = \max_{\rho_{R_mA_m}} \lrV{\cM\lrp{\rho_{R_mA_m}} - \Delta_m\lrp{\rho_{R_mA_m}}}_1,
    \end{align}
    where $R_m$ is an m-dimensional quantum system and $\rho_{R_mA_m}$ is a quantum state.
    Denote $1 - E'_{XY} \star J^{\vec{\cC}}_{XY}$ by $\alpha$.
    For any quantum state $\rho_{R_mA_m}$, consider that
    \begin{align}
        \cM\lrp{\rho_{R_mA_m}} - \Delta_m\lrp{\rho_{R_mA_m}} &= \lrp{J^\cM_{A_mB_m} - J^{\Delta_m}_{A_mB_m}} \star \rho_{R_mA_m}\\
        &= \frac{\alpha}{m-1} \lrp{\idop_{A_mB_m} \star \rho_{R_mA_m} - m J^{\Delta_m}_{A_mB_m} \star \rho_{R_mA_m}}\\
        &= \frac{\alpha}{m-1} \lrp{\rho_{R_m} \ox \idop_{B_m} - m \tr_{A_m}\lrb{\lrp{\sum_{j=0}^{m-1} \idop_{R_m} \ox \proj{j}_{A_m} \ox \proj{j}_{B_m}} \lrp{\rho_{R_mA_m} \ox \idop_{B_m}}}}\\
        &= \frac{\alpha}{m-1} \lrp{\sum_{j=0}^{m-1} \rho_{R_m} \ox \proj{j}_{B_m} - m \sum_{j=0}^{m-1}\tr_{A_m}\lrb{\lrp{\idop_{R_m} \ox \proj{j}_{A_m}} \rho_{R_mA_m}} \ox \proj{j}_{B_m}}\\
        &= \frac{\alpha}{m-1} \sum_{j=0}^{m-1}\lrp{\rho_{R_m} - m \tr_{A_m}\lrb{\lrp{\idop_{R_m} \ox \proj{j}_{A_m}} \rho_{R_mA_m}}} \ox \proj{j}_{B_m}\\
        &= \frac{\alpha}{m-1} \sum_{j=0}^{m-1}\lrp{\rho_{R_m} - m\widetilde{\rho}^j_{R_m}} \ox \proj{j}_{B_m}\\
        &= \frac{\alpha}{m-1} \sum_{j=0}^{m-1}\lrp{\rho_{R_m} - \widetilde{\rho}^j_{R_m} - (m - 1) \widetilde{\rho}^j_{R_m}} \ox \proj{j}_{B_m},
    \end{align}
    where we define $\widetilde{\rho}^j_{R_m} \coloneqq \tr_{A_m}\lrb{\lrp{\idop_{R_m} \ox \proj{j}_{A_m}} \rho_{R_mA_m}}$.
    Since $\rho_{R_m} - \widetilde{\rho}^j_{R_m} = \sum_{k=0,\,k\neq j}^{m-1}\widetilde{\rho}^k_{R_m}$ and $\widetilde{\rho}^j_{R_m}$ are both positive semidefinite, we have $\lrV{\rho_{R_m} - \widetilde{\rho}^j_{R_m}}_1 = \tr\lrb{\rho_{R_m} - \widetilde{\rho}^j_{R_m}} = 1 - \tr\lrb{\widetilde{\rho}^j_{R_m}}$ and $\lrV{\widetilde{\rho}^j_{R_m}}_1 = \tr \lrb{\widetilde{\rho}^j_{R_m}}$.
    Then,
    \begin{align}
        \lrV{\cM\lrp{\rho_{R_mA_m}} - \Delta_m\lrp{\rho_{R_mA_m}}}_1 &= \lrV{\frac{\alpha}{m-1} \sum_{j=0}^{m-1}\lrp{\rho_{R_m} - \widetilde{\rho}^j_{R_m} - (m - 1) \widetilde{\rho}^j_{R_m}} \ox \proj{j}_{B_m}}_1\\
        \label{eq:trace_dist_expression}
        &= \frac{\alpha}{m-1} \sum_{j=0}^{m-1}\lrV{\rho_{R_m} - \widetilde{\rho}^j_{R_m} - (m - 1) \widetilde{\rho}^j_{R_m}}_1\\
        &\leq \frac{\alpha}{m-1} \sum_{j=0}^{m-1}\lrp{\lrV{\rho_{R_m} - \widetilde{\rho}^j_{R_m}}_1 + (m - 1)\lrV{\widetilde{\rho}^j_{R_m}}_1}\\
        &= \frac{\alpha}{m-1} \sum_{j=0}^{m-1}\lrp{1 + (m - 2)\tr\lrb{\widetilde{\rho}^j_{R_m}}}\\
        &= \frac{\alpha}{m-1} \lrp{m + (m - 2)\tr\lrb{\rho_{R_m}}}\\
        &= 2\alpha.
    \end{align}
    Since this holds for any state $\rho_{R_mA_m}$, we obtain an upper bound on the diamond distance:
    \begin{align}
        \frac{1}{2}\lrV{\cM - \Delta_m}_\diamond = \frac{1}{2}\max_{\rho_{R_mA_m}} \lrV{\cM\lrp{\rho_{R_mA_m}} - \Delta_m\lrp{\rho_{R_mA_m}}}_1 \leq \alpha.
    \end{align}
    On the other hand, $\frac{1}{2}\lrV{\cM - \Delta_m}_\diamond$ is lower bounded by $\alpha$ because, for $\rho_{R_mA_m} = \proj{0}_{R_m} \ox \proj{0}_{A_m}$, $\rho_{R_m} = \widetilde{\rho}^0_{R_m} = \proj{0}_{R_m}$ and $\widetilde{\rho}^j_{R_m} = 0$ for all $j\neq 0$, Eq.~\eqref{eq:trace_dist_expression} gives
    \begin{align}
        \lrV{\cM\lrp{\rho_{R_mA_m}} - \Delta_m\lrp{\rho_{R_mA_m}}}_1 &= \frac{\alpha}{m-1} \lrp{\lrV{-(m - 1) \widetilde{\rho}^0_{R_m}}_1 + (m - 1)\lrV{\rho_{R_m}}_1}\\
        &= 2\alpha.
    \end{align}
    Since $\alpha$ is both an upper bound and a lower bound on $\frac{1}{2}\lrV{\cM - \Delta_m}_\diamond$, we obtain
    \begin{align}
        \frac{1}{2}\lrV{\cM - \Delta_m}_\diamond = \alpha = 1 - E'_{XY} \star J^{\vec{\cC}}_{XY}.
    \end{align}
    Therefore, we can write the optimization program of the minimum error as
    \begin{align}
    &\begin{aligned}\label{appeq:min_diamond_opt}
        \omega^{\bf X}(\vec{\cC}, \Delta_m) = \min &\; 1 - E'_{XY} \star J^{\vec{\cC}}_{XY}\\
        \text{\rm s.t.} &\; E'_{XY} \ox J^{\Delta_m}_{A_mB_m} + F'_{XY} \ox \lrp{\idop_{A_mB_m}-J^{\Delta_m}_{A_mB_m}} \in \widetilde{\bf X}_{A_mB_m},
    \end{aligned}
    \end{align}
    {\em i.e.}, $\omega^{\bf X}(\vec{\cC}, \Delta_m)$ is obtained from the same optimization program as $p^{\bf X}(\vec{\cC}, m)$ in Eq.~\eqref{supp_eq:optim_pb_pX}, which implies that the two quantities are the same: $\omega^{\bf X}(\vec{\cC}, \Delta_m) = p^{\bf X}(\vec{\cC}, m)$, as we set out to prove.
\end{proof}

%%%%%%%%%%%%%%%%%%%%%%%%%%%%%%%%%%%%%%%%%%%%%%%%%%%%%%%%%%%%%%%%%%%%%%%%%%%
\subsection{SDP for one-shot classical capacities assisted by free general supermaps}
\begin{proposition}\label{prop:gen_capacity_sdp}
    Given a list of $N$ quantum channels $\vec{\cC}$, its one-shot classical capacity assisted by free general supermaps with error tolerance $\epsilon \geq 0$ can be formulated as
    \begin{align}
    \begin{aligned}\label{supp_eq:CFree_SDP}
        C_\epsilon^{\bf Free} (\,\vec{\cC}\,) = \log_2 \max &\; \left\lfloor m\right\rfloor\\
        \text{\rm s.t.} &\; E_{XY} \star J^{\vec{\cC}}_{XY} \ge m(1 - \epsilon),\, E_{Y} = \idop_{Y},\\
        &\; 0 \leq E_{XY} \leq F_{XY},\, \cL^{\rm NS}_{XY}\lrp{F_{XY}} = \frac{m}{d_{X}}\idop_{XY},
    \end{aligned}
    \end{align}
    where $m$ is maximized over real numbers, and $E_{XY}$ and $F_{XY}$ are optimization variables.
\end{proposition}
\begin{proof}
    Recall from Eq.~\eqref{eq:capacity_def} in the Methods that
    \begin{align}
    \begin{aligned}
        C_\epsilon^{\bf X} (\,\vec{\cC}\,) \coloneqq \log_2 \max &\; m\\
        \text{\rm s.t.} &\; \omega^{\bf X}(\vec{\cC}, \Delta_m) \le \epsilon,
    \end{aligned}
    \end{align}
    where $m$ is maximized over all positive integers.
    According to Eq.~\eqref{appeq:min_diamond_opt}, we can write $C_\epsilon^{\bf X} (\,\vec{\cC}\,)$ as
    \begin{align}
    \begin{aligned}\label{supp_eq:omega_capacity_program}
        C_\epsilon^{\bf X} (\,\vec{\cC}\,) = \log_2 \max &\; m\\
        \text{\rm s.t.} &\; 1 - E'_{XY} \star J^{\vec{\cC}}_{XY} \le \epsilon,\\
        &\; E'_{XY} \ox J_{A_mB_m}^{\Delta_m} + F'_{XY} \ox \lrp{\idop_{A_mB_m} - J_{A_mB_m}^{\Delta_m}} \in \widetilde{\bf X}.
    \end{aligned}
    \end{align}
    Choi operators of signaling-non-generating supermaps should satisfy constraint~\eqref{supp_eq:forward_no_sig_init}, which is
    \begin{align}\label{supp_eq:forward_no_sig}
        {}_{[1-A_m]}J_{YA_mB_m} = 0
    \end{align}
    for the Choi operator $J_{XYA_mB_m}$ of a supermap.
    Applying this constraint to $E'_{XY}\ox J_{A_mB_m}^{\Delta_m} + F'_{XY}\ox\lrp{\idop_{A_mB_m}-J_{A_mB_m}^{\Delta_m}}$ results in
    \begin{align}
        {}_{[1-A_m]}\lrp{E'_Y \ox J^{\Delta_m}_{A_mB_m} + F'_Y \ox \lrp{\idop_{A_mB_m}-J^{\Delta_m}_{A_mB_m}}} &= E'_Y \ox {}_{[1-A_m]}J^{\Delta_m}_{A_mB_m} + F'_Y \ox {}_{[1-A_m]}\lrp{\idop_{A_mB_m} - J^{\Delta_m}_{A_mB_m}}\\
        &= \lrp{E'_Y - F'_Y} \ox \lrp{J^{\Delta_m}_{A_mB_m} - \frac{1}{m}\idop_{A_mB_m}}\\
        &= 0,
    \end{align}
    which is true if and only if $E'_Y = F'_Y$.

    The case where ${\bf X} = {\bf Free}$ corresponds to the constraints given in Proposition~\ref{prop:gen_char}, in addition to the No Forward Signaling Condition~\eqref{supp_eq:forward_no_sig_init}.
    The first constraint in Proposition~\ref{prop:gen_char} requires both $E'_{XY} \geq 0$ and $F'_{XY} \geq 0$.
    The second constraint gives us
    \begin{align}
        \cL^{\rm NS}_{XY}\lrp{E'_{XY} + (m-1)F'_{XY}} = \frac{\idop_{XY}}{d_X},
    \end{align}
    as in Eq.~\eqref{supp_eq:LNS_E_m1F}, which, together with the observation that $\tr_X[\cL^{\rm NS}_{XY}(\cdot)] = \tr_X[\cdot]$, implies
    \begin{align}
        \tr_X\lrb{\cL^{\rm NS}_{XY}\lrp{E'_{XY} + (m-1)F'_{XY}}} = E'_Y + (m-1)F'_Y = \idop_Y.
    \end{align}
    Under the above equality, the No Forward Signaling Condition $E'_Y = F'_Y$ is equivalent to $mE'_Y = \idop_Y$.
    Thus, we obtain
    \begin{align}
    \begin{aligned}
        C_\epsilon^{\bf Free} (\,\vec{\cC}\,) = \log_2 \max &\; m\\
        \text{\rm s.t.} &\; 1 - E'_{XY} \star J^{\vec{\cC}}_{XY} \le \epsilon,\, mE'_Y = \idop_Y,\\
        &\; E'_{XY} \geq 0,\, F'_{XY} \geq 0,\, \cL^{\rm NS}_{XY}\lrp{E'_{XY} + (m-1)F'_{XY}} = \frac{\idop_{XY}}{d_X}.
    \end{aligned}
    \end{align}
    Relabeling $mE'_{XY}$ as $E_{XY}$ and $mE'_{XY} + m(m-1)F'_{XY}$ as $F_{XY}$ leads to
    \begin{align}
    \begin{aligned}
        C_\epsilon^{\bf Free} (\,\vec{\cC}\,) = \log_2 \max &\; m\\
        \text{\rm s.t.} &\; E_{XY} \star J^{\vec{\cC}}_{XY} \ge m(1 - \epsilon),\, E_Y = \idop_Y,\\
        &\; 0 \leq E_{XY} \leq F_{XY},\, \cL^{\rm NS}_{XY}\lrp{F_{XY}} = \frac{m}{d_X}\idop_{XY}.
    \end{aligned}
    \end{align}
    While $m$ is maximized over integers, one can actually maximize it over real numbers and then take the largest integer smaller than it.
    To see this, suppose $E^*_{XY}$, $F^*_{XY}$, $m^*$ are optimal when $m$ is maximized over real numbers, and $m^*$ is not an integer.
    Then, denoting $p \coloneqq \frac{\left\lfloor m^* \right\rfloor - 1}{m^* - 1}$, one can verify that $p E^*_{XY} + \frac{1-p}{d_X}\idop_{XY}$, $p F^*_{XY} + \frac{1-p}{d_X}\idop_{XY}$, and $\left\lfloor m^* \right\rfloor$ are feasible.
    This allows us to replace $m$ by $\lfloor m \rfloor$ in the optimization program above, as in Eq.~\eqref{supp_eq:CFree_SDP}, and thus shows that $C_\epsilon^{\bf Free} (\,\vec{\cC}\,)$ can be formulated as an SDP.
\end{proof}

%%%%%%%%%%%%%%%%%%%%%%%%%%%%%%%%%%%%%%%%%%%%%%%%%%%%%%%%%%%%%%%%%%%%%%%%%%%
\subsection{SDP for one-shot classical capacities assisted by free causally definite supermaps}
In the bipartite case, causally definite supermaps are equivalent to QC-CC operations, and the one-shot classical capacity assisted by free causally definite supermaps can be formulated as an SDP following a similar procedure.

\begin{proposition}
    Given a list of two quantum channels $\vec{C}$, its one-shot classical capacity assisted by free causally definite supermaps with error tolerance $\epsilon \geq 0$ can be formulated as
    \begin{align}
    \begin{aligned}
        C_\epsilon^{\bf FreeDef} (\,\vec{\cC}\,) = \log_2 \max &\; \left\lfloor m\right\rfloor\\
        \text{\rm s.t.} &\; \lrp{E_{XY} + G_{XY}} \star J^{\vec{\cC}}_{XY} \ge m(1 - \epsilon),\, E_Y + G_Y = \idop_Y,\\
        &\; 0 \leq E_{XY} \leq F_{XY},\, 0 \leq G_{XY} \leq H_{XY},\, \tr\lrb{F_{XY} + H_{XY}} = md_Y,\\
        &\; {}_{\lrb{1-Y_2}}F_{XY} = 0,\, {}_{\lrb{1-Y_1}}F_{X_1Y_1} = 0,\, {}_{\lrb{1-Y_1}}H_{XY} = 0,\, {}_{\lrb{1-Y_2}}H_{X_2Y_2} = 0,
    \end{aligned}
    \end{align}
    where $m$ is maximized over real numbers, and $E_{XY}$, $F_{XY}$, $G_{XY}$, and $H_{XY}$ are optimization variables.
\end{proposition}
\begin{proof}
    In the bipartite case, the constraints given in Proposition~\ref{prop:def_char} on the Choi operator $J_{XYA_mB_m}$ of a causally definite supermap reduce to
    \begin{align}
    \begin{gathered}\label{appeq:def_constraints_bipartite}
        J_{XYA_mB_m} = J^{(1,2)}_{XYA_mB_m} + J^{(2,1)}_{XYA_mB_m},\quad J^{(1,2)}_{A_m} + J^{(2,1)}_{A_m} = d_Y\idop_{A_m},\\
        J^{(1,2)}_{XYA_mB_m} \geq 0,\quad J^{(2,1)}_{XYA_mB_m} \geq 0,\\
        {}_{\lrb{1-Y_1}}J^{(1,2)}_{X_1Y_1A_m} = 0,\quad {}_{\lrb{1-Y_2}}J^{(2,1)}_{X_2Y_2A_m} = 0,\\
        {}_{\lrb{1-Y_2}}J^{(1,2)}_{XYA_m} = 0,\quad {}_{\lrb{1-Y_1}}J^{(2,1)}_{XYA_m} = 0.
    \end{gathered}
    \end{align}
    Since simultaneous permutation at $A_m$ and $B_m$ on $J_{XYA_mB_m}$ does not change the optimal value when solving the optimization program of the minimum error, we can restrict $J^{(1,2)}_{XYA_mB_m}$ and $J^{(2,1)}_{XYA_mB_m}$ to operators of the form
    \begin{align}
    \begin{gathered}
        J^{(1,2)}_{XYA_mB_m} = E'_{XY} \ox J^{\Delta_m}_{A_mB_m} + F'_{XY} \ox \lrp{\idop_{A_mB_m} - J^{\Delta_m}_{A_mB_m}},\\
        J^{(2,1)}_{XYA_mB_m} = G'_{XY} \ox J^{\Delta_m}_{A_mB_m} + H'_{XY} \ox \lrp{\idop_{A_mB_m} - J^{\Delta_m}_{A_mB_m}},
    \end{gathered}
    \end{align}
    where $E'_{XY}$, $F'_{XY}$, $G'_{XY}$, and $H'_{XY}$ are arbitrary Hermitian operators.
    Then, we can rewrite constraints~\eqref{appeq:def_constraints_bipartite} in terms of $E'_{XY}$, $F'_{XY}$, $G'_{XY}$, and $H'_{XY}$.
    
    The constraint $J^{(1,2)}_{A_m} + J^{(2,1)}_{A_m} = d_Y\idop_{A_m}$ is equivalent to
    \begin{align}
        \tr\lrb{E'_{XY} + (m-1)F'_{XY} + G'_{XY} + (m-1)H'_{XY}} = d_Y.
    \end{align}
    The constraints $J^{(1,2)}_{XYA_mB_m} \geq 0$ and $J^{(2,1)}_{XYA_mB_m} \geq 0$ correspond to
    \begin{align}
        E'_{XY} \geq 0,\quad F'_{XY} \geq 0,\quad G'_{XY} \geq 0,\quad \text{and}\quad H'_{XY} \geq 0.
    \end{align}
    Since $J^{(1,2)}_{XYA_m} = \lrp{E'_{XY} + (m-1)F'_{XY}} \ox \idop_{A_m}$, the constraints ${}_{\lrb{1-Y_2}}J^{(1,2)}_{XYA_m} = 0$ and ${}_{\lrb{1-Y_1}}J^{(1,2)}_{X_1Y_1A_m} = 0$ correspond to
    \begin{align}
        {}_{\lrb{1-Y_2}}\lrp{E'_{XY} + (m-1)F'_{XY}} = 0 \quad\text{and}\quad {}_{\lrb{1-Y_1}}\lrp{E'_{X_1Y_1} + (m-1)F'_{X_1Y_1}} = 0.
    \end{align}
    Similarly, ${}_{\lrb{1-Y_1}}J^{(2,1)}_{XYA_m} = 0$ and ${}_{\lrb{1-Y_2}}J^{(2,1)}_{X_2Y_2A_m} = 0$ correspond to
    \begin{align}
        {}_{\lrb{1-Y_1}}\lrp{G'_{XY} + (m-1)H'_{XY}} = 0 \quad\text{and}\quad {}_{\lrb{1-Y_2}}\lrp{G'_{X_2Y_2} + (m-1)H'_{X_2Y_2}} = 0.
    \end{align}
    Finally, comparing $J_{XYA_mB_m} = \lrp{E'_{XY} + G'_{XY}} \ox J^{\Delta_m}_{A_mB_m} + ({F'_{XY} + H'_{XY}}) \ox \lrp{\idop_{A_mB_m} - J^{\Delta_m}_{A_mB_m}}$ with Eq.~\eqref{supp_eq:omega_capacity_program}, we note that $E'_{XY}$ and $F'_{XY}$ in Eq.~\eqref{supp_eq:omega_capacity_program} correspond to $E'_{XY} + G'_{XY}$ and $F'_{XY} + H'_{XY}$, respectively, in this case.
    Hence, we arrive at
    \begin{align}
    \begin{aligned}
        C_\epsilon^{\bf FreeDef} (\,\vec{\cC}\,) = \log_2 \max &\; m\\
        \text{\rm s.t.} &\; \lrp{E'_{XY} + G'_{XY}} \star J^{\vec{\cC}}_{XY} \ge 1 - \epsilon,\, E'_Y + G'_Y = F'_Y + H'_Y,\\
        &\; \tr\lrb{E'_{XY} + (m-1)F'_{XY} + G'_{XY} + (m-1)H'_{XY}} = d_Y,\\
        &\; E'_{XY} \geq 0,\, F'_{XY} \geq 0,\, G'_{XY} \geq 0,\, H'_{XY} \geq 0,\\
        &\; {}_{\lrb{1-Y_2}}\lrp{E'_{XY} + (m-1)F'_{XY}} = 0,\, {}_{\lrb{1-Y_1}}\lrp{E'_{X_1Y_1} + (m-1)F'_{X_1Y_1}} = 0,\\
        &\; {}_{\lrb{1-Y_1}}\lrp{G'_{XY} + (m-1)H'_{XY}} = 0,\, {}_{\lrb{1-Y_2}}\lrp{G'_{X_2Y_2} + (m-1)H'_{X_2Y_2}} = 0.
    \end{aligned}
    \end{align}
    Since $J_{XYAB}$ also satisfy the causality constraint $\cL^{\rm NS}_{XY}\lrp{J_{XYA}} = \frac{\idop_{XYA}}{d_X}$ for a general supermap, we can equivalently write the No Forward Signaling Condition $E'_Y + G'_Y = F'_Y + H'_Y$ as
    \begin{align}
        m\lrp{E'_Y + G'_Y} = \idop_Y,
    \end{align}
    as we did in the case with general supermaps.
    Relabeling $mE'_{XY}$ as $E_{XY}$, $mG'_{XY}$ as $G_{XY}$, $mE'_{XY} + m(m-1)F'_{XY}$ as $F_{XY}$, and $mG'_{XY} + m(m-1)H'_{XY}$ as $H_{XY}$, we obtain
    \begin{align}
    \begin{aligned}
        C_\epsilon^{\bf FreeDef} (\,\vec{\cC}\,) = \log_2 \max &\; m\\
        \text{\rm s.t.} &\; \lrp{E_{XY} + G_{XY}} \star J^{\vec{\cC}}_{XY} \ge m(1 - \epsilon),\, E_Y + G_Y = \idop_Y,\\
        &\; 0 \leq E_{XY} \leq F_{XY},\, 0 \leq G_{XY} \leq H_{XY},\, \tr\lrb{F_{XY} + H_{XY}} = md_Y,\\
        &\; {}_{\lrb{1-Y_2}}F_{XY} = 0,\, {}_{\lrb{1-Y_1}}F_{X_1Y_1} = 0,\, {}_{\lrb{1-Y_1}}H_{XY} = 0,\, {}_{\lrb{1-Y_2}}H_{X_2Y_2} = 0.
    \end{aligned}
    \end{align}
    By the same argument as in the case with general supermaps, we can relax $m$ to be maximized over real numbers and then take the largest integer smaller than it, which completes the proof.
\end{proof}

%%%%%%%%%%%%%%%%%%%%%%%%%%%%%%%%%%%%%%%%%%%%%%%%%%%%%%%%%%%%%%%%%%%%%%%%%%%
\subsection{SDP for one-shot classical capacities assisted by free parallel supermaps}
The one-shot classical capacity assisted by free parallel supermaps is the same as that assisted by no-signaling correlations in the standard communication scenario, which can be formulated as the following SDP~\cite{leung2015power,wang2017semidefinite}:
\begin{align}
\begin{aligned}\label{appeq:par_capacity_sdp}
    C_\epsilon^{\bf FreePar} (\,\vec{\cC}\,) = \log_2 \max &\; \left\lfloor m\right\rfloor\\
    \text{\rm s.t.} &\; E_{XY} \star J^{\vec{\cC}}_{XY} \ge m(1 - \epsilon),\, E_Y = \idop_Y,\\
    &\; 0 \leq E_{XY} \leq F_X \ox \idop_Y,\, \tr\lrb{F_X} = m,
\end{aligned}
\end{align}
where $m$ is maximized over real numbers, and $E_{XY}$ and $F_X$ are optimization variables.

%%%%%%%%%%%%%%%%%%%%%%%%%%%%%%%%%%%%%%%%%%%%%%%%%%%%%%%%%%%%%%%%%%%%%%%%%%%
%%%%%%%%%%%%%%%%%%%%%%%%%%%%%%%%%%%%%%%%%%%%%%%%%%%%%%%%%%%%%%%%%%%%%%%%%%%
\section{Advantage of indefinite causal order in zero-error communication}\label{supp_sec:advantage}
\begin{corollary}[SDPs for one-shot zero-error classical capacity assisted by free general supermaps]
    Given a list of $N$ quantum channels $\vec{\cC}$, its one-shot zero-error classical capacity assisted by free general supermaps can be formulated as
    \begin{align}
    &\begin{aligned}\label{eq:classical_capacity_gen_sdp_primal}
        C_0^{\bf Free} (\,\vec{\cC}\,) = \log_2\max &\; \left\lfloor m\right\rfloor\\
        \text{\rm s.t.} &\; E_{XY} \star J^{\vec{\cC}}_{XY} = m,\, 0 \leq E_{XY} \leq F_{XY},\, \cL^{\rm NS}_{XY}(F_{XY}) = \frac{m}{d_X}\idop_{XY},\, E_Y = \idop_Y
    \end{aligned}\\
    &\phantom{C_0^{\bf Free} (\,\vec{\cC}\,)}
    \begin{aligned}\label{eq:classical_capacity_gen_sdp_dual}
        \mathllap{} \leq \log_2\min &\; \left\lfloor \tr\lrb{N_{Y}} \right\rfloor\\
        \text{\rm s.t.} &\; \cL^{\rm NS}_{XY}\lrp{M_{XY}} \geq 0,\, \cL^{\rm NS}_{XY}\lrp{M_{XY}} + \idop_X \ox N_Y \geq \lrp{1 + \frac{\tr\lrb{M_{XY}}}{d_X}} J^{\vec{\cC}}_{XY}
    \end{aligned}
    \end{align}
    with primal variables $m, E_{XY}, F_{XY}$ and dual variables $M_{XY}, N_Y$.
\end{corollary}
\begin{proof}
    Since zero-error is a special case where $\epsilon = 0$, we can obtain an SDP of $C_0^{\bf Free} (\,\vec{\cC}\,)$ from Proposition~\ref{prop:gen_capacity_sdp}:
    \begin{align}
    \begin{aligned}
        C_0^{\bf Free} (\,\vec{\cC}\,) = \log_2\max &\; \left\lfloor m\right\rfloor\\
        \text{\rm s.t.} &\; E_{XY} \star J^{\vec{\cC}}_{XY} \geq m,\, 0 \leq E_{XY} \leq F_{XY},\, \cL^{\rm NS}_{XY}(F_{XY}) = \frac{m}{d_X}\idop_{XY},\, E_Y = \idop_Y.
    \end{aligned}
    \end{align}
    The inequality constraint $E_{XY} \star J^{\vec{\cC}}_{XY} \geq m$ can be changed to the equality constraint $E_{XY} \star J^{\vec{\cC}}_{XY} = m$.
    To see this, suppose $m' \coloneqq E_{XY} \star J^{\vec{\cC}}_{XY} > m$.
    Then, we can replace $F_{XY}$ by $F_{XY} + \frac{m' - m}{d_X}\idop_{XY}$ so that $m'$, $E_{XY}$, and $F_{XY} + \frac{m' - m}{d_X}\idop_{XY}$ form a feasible solution.
    Since this works for every $m$, we arrive at the claimed SDP in Eq.~\eqref{eq:classical_capacity_gen_sdp_primal}.

    To derive the dual SDP~\eqref{eq:classical_capacity_gen_sdp_dual}, we first note that
    \begin{align}
    \begin{aligned}
        -C_0^{\bf Free} (\,\vec{\cC}\,) = \log_2 \min &\; \left\lceil -m\right\rceil\\
        \text{\rm s.t.} &\; E_{XY} \star J^{\vec{\cC}}_{XY} = m,\, 0 \leq E_{XY} \leq F_{XY},\, \cL^{\rm NS}_{XY}(F_{XY}) = \frac{m}{d_X}\idop_{XY},\, E_Y = \idop_Y.
    \end{aligned}
    \end{align}
    The Lagrangian associated with the above SDP is
    \begin{align}
    &\begin{aligned}
        & L\lrp{m, E_{XY}, F_{XY}, \lambda, K^{(1)}_{XY}, K^{(2)}_{XY}, K^{(3)}_{XY}, K^{(4)}_Y}\\
        \coloneqq \,& -m + \lambda\lrp{E_{XY} \star J^{\vec{\cC}}_{XY} - m} - \lrang{K^{(1)}_{XY}, E_{XY}} + \lrang{K^{(2)}_{XY}, E_{XY} - F_{XY}}\\
        &\quad + \lrang{K^{(3)}_{XY}, \cL^{\rm NS}_{XY}(F_{XY}) - \frac{m}{d_X}\idop_{XY}} + \lrang{K^{(4)}_Y, E_Y - \idop_Y}
    \end{aligned}\\
    &\begin{aligned}
        = \,& -m\lrp{1 + \lambda + \frac{1}{d_X}\tr\lrb{K^{(3)}_{XY}}} + \lrang{\lambda \lrp{J^{\vec{\cC}}_{XY}}^\tp - K^{(1)}_{XY} + K^{(2)}_{XY} + \idop_X \ox K^{(4)}_Y, E_{XY}}\\
        &\quad + \lrang{\cL^{\rm NS}_{XY}\lrp{K^{(3)}_{XY}} - K^{(2)}_{XY}, F_{XY}} - \tr\lrb{K^{(4)}_Y},
    \end{aligned}
    \end{align}
    where $\lambda, K^{(1)}_{XY}, K^{(2)}_{XY}, K^{(3)}_{XY}, K^{(4)}_Y$ are dual variables. For the corresponding Lagrange dual function
    \begin{align}
        g\lrp{\lambda, K^{(1)}_{XY}, K^{(2)}_{XY}, K^{(3)}_{XY}, K^{(4)}_Y} \coloneqq \inf_{m, E_{XY}, F_{XY}} L\lrp{m, E_{XY}, F_{XY}, \lambda, K^{(1)}_{XY}, K^{(2)}_{XY}, K^{(3)}_{XY}, K^{(4)}_Y}
    \end{align}
    to be bounded, the dual variables must satisfy
    \begin{align}
    \begin{gathered}
        K^{(1)}_{XY} \geq 0,\quad K^{(2)}_{XY} \geq 0,\quad 1 + \frac{1}{d_X}\tr\lrb{K^{(3)}_{XY}} = -\lambda,\\
        \lambda \lrp{J^{\vec{\cC}}_{XY}}^\tp + K^{(2)}_{XY} + \idop_X \ox K^{(4)}_Y = K^{(1)}_{XY},\quad \cL^{\rm NS}_{XY}\lrp{K^{(3)}_{XY}} = K^{(2)}_{XY}.
    \end{gathered}
    \end{align}
    Note that variables $\lambda$ ,$K^{(1)}_{XY}$ and $K^{(2)}_{XY}$ can be removed from the above constraints, resulting in the simplified constraints
    \begin{align}
    \begin{gathered}
         \cL^{\rm NS}_{XY}\lrp{K^{(3)}_{XY}} + \idop_X \ox K^{(4)}_Y \geq \lrp{1 + \frac{\tr\lrb{K^{(3)}_{XY}}}{d_X}} \lrp{J^{\vec{\cC}}_{XY}}^\tp,\quad \cL^{\rm NS}_{XY}\lrp{K^{(3)}_{XY}} \geq 0.
    \end{gathered}
    \end{align}
    Defining $M_{XY} \coloneqq K^{(3)}_{XY}$ and $N_Y \coloneqq K^{(4)}_Y$, we arrive at the dual SDP of maximizing Lagrange dual function $g$ with the above constraints:
    \begin{align}
    \begin{aligned}
        -C_0^{\bf Free} (\,\vec{\cC}\,) \geq \log_2 \max &\; \left\lceil -\tr\lrb{N_Y}\right\rceil\\
        \text{\rm s.t.} &\; \cL^{\rm NS}_{XY}\lrp{M_{XY}} \geq 0,\, \cL^{\rm NS}_{XY}\lrp{M_{XY}} + \idop_X \ox N_Y \geq \lrp{1 + \frac{\tr\lrb{M_{XY}}}{d_X}} \lrp{J^{\vec{\cC}}_{XY}}^\tp,
    \end{aligned}
    \end{align}
    where the inequality is due to weak duality. Then, $C_0^{\bf Free} (\,\vec{\cC}\,) \leq \log_2 \min \; \left\lfloor \tr\lrb{N_Y}\right\rfloor$ with the same constraints, resulting in the claimed SDP in Eq.~\eqref{eq:classical_capacity_gen_sdp_dual} with the additional observation that dropping the transpose acting on $J^{\vec{\cC}}_{XY}$ does not affect the optimization result.
\end{proof}

Following similar steps, we can derive the SDP for the capacity assisted by causally definite supermaps and its dual.

\begin{corollary}[SDPs for one-shot zero-error classical capacity assisted by free causally definite supermaps]
    Given a list of two quantum channels $\vec{\cC}$, its one-shot zero-error classical capacity assisted by free causally definite supermaps can be formulated as
    \begin{align}
    &\begin{aligned}\label{eq:classical_capacity_def_sdp_primal}
        C_0^{\bf FreeDef} (\,\vec{\cC}\,) = \log_2\max & \; \floor{m} \\
        \text{\rm s.t.} &\; \lrp{E_{XY} + G_{XY}} \star J^{\vec{\cC}}_{XY} = m,\, 0 \leq E_{XY} \leq F_{XY},\, 0 \leq G_{XY} \leq H_{XY},\\
        &\; {}_{[1-Y_2]}F_{XY}=0,\, {}_{[1-Y_1]}F_{X_1Y_1}=0,\, {}_{[1-Y_1]}H_{XY}=0,\, _{[1-Y_2]}H_{X_2Y_2}=0,\\
        &\; \tr\lrb{F_{XY} + H_{XY}} = md_Y,\, E_Y + G_Y =\idop_Y.
        \end{aligned}\\
        &\phantom{C_0^{\bf FreeDef} (\,\vec{\cC}\,)}
        \begin{aligned}\label{eq:classical_capacity_def_sdp_dual}
            \mathllap{} \leq \log_2\min &\; \left\lfloor \tr\lrb{K_Y} \right\rfloor\\
            \text{\rm s.t.} &\; M_{XY} \geq 0,\, _{[1-X_2]}M_{XY_1} = 0,\, M_{X_1} = d_{X_2}\lambda\idop_{X_1},\\
            &\; M_{XY} + \idop_X \ox K_Y \geq (\lambda + 1)J^{\vec{\cC}}_{XY},\\
            &\; N_{XY} \geq 0,\, _{[1-X_1]}N_{XY_2} = 0,\, N_{X_2} = d_{X_1}\lambda\idop_{X_2},\\
            &\; N_{XY} + \idop_X \ox K_Y \geq (\lambda + 1)J^{\vec{\cC}}_{XY}
        \end{aligned}
    \end{align}
    with primal variables $m, E_{XY}, F_{XY}, G_{XY}, H_{XY}$ and dual variables $\lambda, M_{XY}, N_{XY}, K_Y$.
\end{corollary}

\begin{theorem}
    Let $\vec{\cC} \coloneqq (\cA, \cA)$ be a pair of identical amplitude damping channels $\cA$ with damping rates $0.1$.
    The one-shot zero-error classical capacity of $\vec{\cC}$ assisted by causally indefinite supermaps is strictly larger than the capacity assisted by causally definite supermaps:
    \begin{align}
        C_0^{\bf Free} (\,\vec{\cC}\,) = 1 > C_0^{\bf FreeDef} (\,\vec{\cC}\,) = 0.
    \end{align}
    Moreover, this capacity value of $C_0^{\bf Free} (\,\vec{\cC}\,) = 1$ can be achieved with a QC-QC operation.
\end{theorem}
\begin{proof}
    We first prove $C_0^{\bf Free} (\,\vec{\cC}\,) = 1$ by showing that $C_0^{\bf Free} (\,\vec{\cC}\,) \geq 1$ using the primal SDP~\eqref{eq:classical_capacity_gen_sdp_primal} and $C_0^{\bf Free} (\,\vec{\cC}\,) \leq 1$ using the dual SDP~\eqref{eq:classical_capacity_gen_sdp_dual}.
    To see that $C_0^{\bf Free} (\,\vec{\cC}\,) \geq 1$, we note that
    \begin{align}
    \begin{gathered}\label{supp_eq:gen_primal_solution}
        F_{X_1Y_1X_2Y_2} = {\rm diag}\lrp{\frac{13}{20}, \frac{3}{5}, \frac{11}{20}, 0, \frac{3}{5}, \frac{11}{20}, \frac{3}{5}, \frac{1}{20}, \frac{11}{20}, \frac{3}{5}, \frac{1}{4}, \frac{4}{5}, 0, \frac{1}{20}, \frac{4}{5}, \frac{27}{20}},\\
        E_{X_1Y_1X_2Y_2} =
        \begin{pmatrix}
            \frac{1}{4} & 0 & 0 & 0 & 0 & 0 & 0 & 0 & 0 & 0 & 0 & 0 & 0 & 0 & 0 & \frac{4}{9}\\
            0 & 0 & 0 & 0 & 0 & 0 & 0 & 0 & 0 & 0 & 0 & 0 & 0 & 0 & 0 & 0\\
            0 & 0 & \frac{1}{4} & 0 & 0 & 0 & 0 & 0 & 0 & 0 & 0 & 0 & 0 & 0 & \frac{1}{\sqrt{10}} & 0\\
            0 & 0 & 0 & 0 & 0 & 0 & 0 & 0 & 0 & 0 & 0 & 0 & 0 & 0 & 0 & 0\\
            0 & 0 & 0 & 0 & 0 & 0 & 0 & 0 & 0 & 0 & 0 & 0 & 0 & 0 & 0 & 0\\
            0 & 0 & 0 & 0 & 0 & \frac{233}{4860} & 0 & 0 & 0 & 0 & 0 & 0 & 0 & 0 & 0 & 0\\
            0 & 0 & 0 & 0 & 0 & 0 & \frac{8}{15} & 0 & 0 & 0 & 0 & 0 & 0 & 0 & 0 & 0\\
            0 & 0 & 0 & 0 & 0 & 0 & 0 & \frac{233}{4860} & 0 & 0 & 0 & 0 & 0 & 0 & 0 & 0\\
            0 & 0 & 0 & 0 & 0 & 0 & 0 & 0 & \frac{1}{4} & 0 & 0 & \frac{1}{\sqrt{10}} & 0 & 0 & 0 & 0\\
            0 & 0 & 0 & 0 & 0 & 0 & 0 & 0 & 0 & \frac{8}{15} & 0 & 0 & 0 & 0 & 0 & 0\\
            0 & 0 & 0 & 0 & 0 & 0 & 0 & 0 & 0 & 0 & \frac{1}{4} & 0 & 0 & 0 & 0 & 0\\
            0 & 0 & 0 & 0 & 0 & 0 & 0 & 0 & \frac{1}{\sqrt{10}} & 0 & 0 & \frac{7}{15} & 0 & 0 & 0 & 0\\
            0 & 0 & 0 & 0 & 0 & 0 & 0 & 0 & 0 & 0 & 0 & 0 & 0 & 0 & 0 & 0\\
            0 & 0 & 0 & 0 & 0 & 0 & 0 & 0 & 0 & 0 & 0 & 0 & 0 & \frac{233}{4860} & 0 & 0\\
            0 & 0 & \frac{1}{\sqrt{10}} & 0 & 0 & 0 & 0 & 0 & 0 & 0 & 0 & 0 & 0 & 0 & \frac{7}{15} & 0\\
            \frac{4}{9} & 0 & 0 & 0 & 0 & 0 & 0 & 0 & 0 & 0 & 0 & 0 & 0 & 0 & 0 & \frac{1387}{1620}
        \end{pmatrix}
    \end{gathered}
    \end{align}
    form a feasible solution to SDP~\eqref{eq:classical_capacity_gen_sdp_primal} with $m=2$, implying that $C_0^{\bf Free} (\,\vec{\cC}\,) \geq \log_2 2 = 1$.
    To see that $C_0^{\bf Free} (\,\vec{\cC}\,) \leq 1$, we note that
    \begin{align}
    \begin{gathered}
        N_Y = \frac{5}{2}\proj{00},\\
        M_{X_1Y_1X_2Y_2} =
        \begin{pmatrix}
            160 & 0 & 0 & \frac{3808}{25} & 0 & 0 & 0 & 0 & 0 & 0 & 0 & 0 & \frac{3808}{25} & 0 & 0 & \frac{289}{2}\\
            0 & 0 & 0 & 0 & 0 & 0 & 0 & 0 & 0 & 0 & 0 & 0 & 0 & 0 & 0 & 0\\
            0 & 0 & \frac{737}{50} & 0 & 0 & 0 & 0 & 0 & 0 & 0 & 0 & 0 & 0 & 0 & \frac{737}{50} & 0\\
            \frac{3808}{25} & 0 & 0 & \frac{7263}{50} & 0 & 0 & 0 & 0 & 0 & 0 & 0 & 0 & 145 & 0 & 0 & \frac{6879}{50}\\
            0 & 0 & 0 & 0 & 0 & 0 & 0 & 0 & 0 & 0 & 0 & 0 & 0 & 0 & 0 & 0\\
            0 & 0 & 0 & 0 & 0 & 0 & 0 & 0 & 0 & 0 & 0 & 0 & 0 & 0 & 0 & 0\\
            0 & 0 & 0 & 0 & 0 & 0 & 0 & 0 & 0 & 0 & 0 & 0 & 0 & 0 & 0 & 0\\
            0 & 0 & 0 & 0 & 0 & 0 & 0 & 0 & 0 & 0 & 0 & 0 & 0 & 0 & 0 & 0\\
            0 & 0 & 0 & 0 & 0 & 0 & 0 & 0 & \frac{737}{50} & 0 & 0 & \frac{737}{50} & 0 & 0 & 0 & 0\\
            0 & 0 & 0 & 0 & 0 & 0 & 0 & 0 & 0 & 0 & 0 & 0 & 0 & 0 & 0 & 0\\
            0 & 0 & 0 & 0 & 0 & 0 & 0 & 0 & 0 & 0 & 0 & 0 & 0 & 0 & 0 & 0\\
            0 & 0 & 0 & 0 & 0 & 0 & 0 & 0 & \frac{737}{50} & 0 & 0 & \frac{737}{50} & 0 & 0 & 0 & 0\\
            \frac{3808}{25} & 0 & 0 & 145 & 0 & 0 & 0 & 0 & 0 & 0 & 0 & 0 & \frac{7263}{50} & 0 & 0 & \frac{6879}{50}\\
            0 & 0 & 0 & 0 & 0 & 0 & 0 & 0 & 0 & 0 & 0 & 0 & 0 & 0 & 0 & 0\\
            0 & 0 & \frac{737}{50} & 0 & 0 & 0 & 0 & 0 & 0 & 0 & 0 & 0 & 0 & 0 & \frac{737}{50} & 0\\
            \frac{289}{2} & 0 & 0 & \frac{6879}{50} & 0 & 0 & 0 & 0 & 0 & 0 & 0 & 0 & \frac{6879}{50} & 0 & 0 & \frac{3263}{25}
        \end{pmatrix}
    \end{gathered}
    \end{align}
    form a feasible solution to SDP~\eqref{eq:classical_capacity_gen_sdp_dual} with $\tr\lrb{N_Y} = \frac{5}{2}$, implying that $C_0^{\bf Free} (\,\vec{\cC}\,) \leq \log_2 \left\lfloor \frac{5}{2} \right\rfloor = \log_2 2 = 1$. Hence, we conclude that $C_0^{\bf Free} (\,\vec{\cC}\,) = 1$.

    One important thing to note is that the feasible solution given in Eq.~\eqref{supp_eq:gen_primal_solution} actually corresponds to a process that can be realized with QC-QC.
    According to the proof of Proposition~\ref{supp_prop:diamond_prob_equal} and Proposition~\ref{prop:gen_capacity_sdp}, the Choi operator $J_{XYAB}$ of a supermap in ${\bf Free}$ that enables the perfect transmission of one classical bit can be constructed from the feasible solution given in Eq.~\eqref{supp_eq:gen_primal_solution} as
    \begin{align}\label{supp_eq:choi_from_feasible}
        J_{XYAB} = \frac{E_{XY}}{m} \ox J^{\Delta_m}_{AB} + \frac{F_{XY} - E_{XY}}{m(m-1)} \ox (\idop_{AB} - J^{\Delta_m}_{AB})
    \end{align}
    with $m = 2$.
    In the case of $N = 2$, the characterization of QC-QC operations given in Proposition~\ref{prop:qcqc_char} reduces to
    \begin{align}
    \begin{gathered}
        J_{XYAB} \geq 0,\quad J^{(\{1\},2)}_{XY_1A} \geq 0,\quad J^{(\{2\},1)}_{XY_2A} \geq 0,\quad J^{(\emptyset,1)}_{X_1A} \geq 0,\quad J^{(\emptyset, 2)}_{X_2A} \geq 0,\\
        J_{XYA} = J^{(\{1\},2)}_{XY_1A} \ox \frac{\idop_{Y_2}}{d_{Y_2}} + J^{(\{2\},1)}_{XY_2A} \ox \frac{\idop_{Y_1}}{d_{Y_1}},\\
        J^{(\{1\},2)}_{X_1Y_1A} = J^{(\emptyset,1)}_{X_1A} \ox \frac{\idop_{Y_1}}{d_{Y_1}},\quad J^{(\{2\},1)}_{X_2Y_2A} = J^{(\emptyset,2)}_{X_2A} \ox \frac{\idop_{Y_2}}{d_{Y_2}},\quad J^{(\emptyset,1)}_{A} + J^{(\emptyset,2)}_{A} = d_Y\idop_A.
    \end{gathered}
    \end{align}
    Continuing from Eq.~\eqref{supp_eq:choi_from_feasible}, we have
    \begin{align}
        J_{XYA} &= \frac{E_{XY}}{m} \ox \idop_A + \frac{F_{XY} - E_{XY}}{m(m-1)} \ox (m - 1)\idop_A\\
        &= F_{XY} \ox \frac{\idop_A}{m}.
    \end{align}
    As a result, to show that $J_{XYAB}$ is the Choi operator of a QC-QC operation is to show that there exist positive semidefinite operators $J^{(\{1\},2)}_{XY_1A}$, $J^{(\{2\},1)}_{XY_2A}$, $J^{(\emptyset,1)}_{X_1A}$, and $J^{(\emptyset,2)}_{X_2A}$ such that
    \begin{align}
    \begin{gathered}
        F_{XY} \ox \frac{\idop_A}{m} = J^{(\{1\},2)}_{XY_1A} \ox \frac{\idop_{Y_2}}{d_{Y_2}} + J^{(\{2\},1)}_{XY_2A} \ox \frac{\idop_{Y_1}}{d_{Y_1}},\\
        J^{(\{1\},2)}_{X_1Y_1A} = J^{(\emptyset,1)}_{X_1A} \ox \frac{\idop_{Y_1}}{d_{Y_1}},\quad J^{(\{2\},1)}_{X_2Y_2A} = J^{(\emptyset,2)}_{X_2A} \ox \frac{\idop_{Y_2}}{d_{Y_2}},\quad J^{(\emptyset,1)}_{A} + J^{(\emptyset,2)}_{A} = d_Y\idop_A.
    \end{gathered}
    \end{align}
    It is straightforward to verify that the following positive semidefinite operators satisfy the above constraints for the Choi operator $J_{XYAB}$ constructed in Eq.~\eqref{supp_eq:choi_from_feasible}:
    \begin{align}
    \begin{gathered}
        J^{(\{1\},2)}_{XY_1A} = J^{(\{1\},2)}_{XY_1} \ox \frac{\idop_A}{m},\quad J^{(\{2\},1)}_{XY_2A} = J^{(\{2\},1)}_{XY_2} \ox \frac{\idop_A}{m},\quad J^{(\emptyset,1)}_{X_1A} = J^{(\emptyset,1)}_{X_1} \ox \frac{\idop_A}{m},\quad J^{(\emptyset,2)}_{X_2A} = J^{(\emptyset,2)}_{X_2} \ox \frac{\idop_A}{m},
    \end{gathered}
    \end{align}
    where
    \begin{align}
    \begin{gathered}
        J^{(\{1\},2)}_{X_1Y_1X_2} \coloneqq {\rm diag}\lrp{\frac{13}{20}, 0, \frac{11}{20}, \frac{1}{10}, \frac{11}{10}, \frac{1}{4}, 0, \frac{27}{20}},\quad J^{(\{2\},1)}_{XY_2} \coloneqq {\rm diag}\lrp{\frac{13}{20}, \frac{11}{20}, \frac{11}{10}, 0, 0, \frac{1}{10}, \frac{1}{4}, \frac{27}{20}},\\
        J^{(\emptyset,1)}_{X_1} = J^{(\emptyset,2)}_{X_2} = {\rm diag}\lrp{\frac{13}{10}, \frac{27}{10}},
    \end{gathered}
    \end{align}
    implying that the feasible $J_{XYAB}$ indeed corresponds to a QC-QC operation.

    Now, we prove $C_0^{\bf FreeDef} (\,\vec{\cC}\,) = 0$ using the dual SDP~\eqref{eq:classical_capacity_def_sdp_dual}.
    We note that
    \begin{align}
    \begin{gathered}
        \lambda = 200,\quad K_Y = \frac{39}{20}\proj{00},\\
        M_{X_1Y_1X_2Y_2} =
        \begin{pmatrix}
            200 & 0 & 0 & \frac{1903}{10} & 0 & 0 & 0 & 0 & 0 & 0 & 0 & 0 & \frac{952}{5} & 0 & 0 & \frac{361}{2}\\
            0 & 0 & 0 & 0 & 0 & 0 & 0 & 0 & 0 & 0 & 0 & 0 & 0 & 0 & 0 & 0\\
            0 & 0 & \frac{189}{10} & 0 & 0 & 0 & 0 & 0 & 0 & 0 & 0 & 0 & 0 & 0 & \frac{931}{50} & 0\\
            \frac{1903}{10} & 0 & 0 & \frac{1811}{10} & 0 & 0 & 0 & 0 & 0 & 0 & 0 & 0 & \frac{4528}{25} & 0 & 0 & \frac{8589}{50}\\
            0 & 0 & 0 & 0 & 0 & 0 & 0 & 0 & 0 & 0 & 0 & 0 & 0 & 0 & 0 & 0\\
            0 & 0 & 0 & 0 & 0 & 0 & 0 & 0 & 0 & 0 & 0 & 0 & 0 & 0 & 0 & 0\\
            0 & 0 & 0 & 0 & 0 & 0 & 0 & 0 & 0 & 0 & 0 & 0 & 0 & 0 & 0 & 0\\
            0 & 0 & 0 & 0 & 0 & 0 & 0 & 0 & 0 & 0 & 0 & 0 & 0 & 0 & 0 & 0\\
            0 & 0 & 0 & 0 & 0 & 0 & 0 & 0 & \frac{373}{20} & 0 & 0 & \frac{3719}{200} & 0 & 0 & 0 & 0\\
            0 & 0 & 0 & 0 & 0 & 0 & 0 & 0 & 0 & 0 & 0 & 0 & 0 & 0 & 0 & 0\\
            0 & 0 & 0 & 0 & 0 & 0 & 0 & 0 & 0 & 0 & \frac{3}{50} & 0 & 0 & 0 & 0 & 0\\
            0 & 0 & 0 & 0 & 0 & 0 & 0 & 0 & \frac{3719}{200} & 0 & 0 & \frac{1859}{100} & 0 & 0 & 0 & 0\\
            \frac{952}{5} & 0 & 0 & \frac{4528}{25} & 0 & 0 & 0 & 0 & 0 & 0 & 0 & 0 & \frac{3627}{20} & 0 & 0 & \frac{4294}{25}\\
            0 & 0 & 0 & 0 & 0 & 0 & 0 & 0 & 0 & 0 & 0 & 0 & 0 & 0 & 0 & 0\\
            0 & 0 & \frac{931}{50} & 0 & 0 & 0 & 0 & 0 & 0 & 0 & 0 & 0 & 0 & 0 & \frac{459}{25} & 0\\
            \frac{361}{2} & 0 & 0 & \frac{8589}{50} & 0 & 0 & 0 & 0 & 0 & 0 & 0 & 0 & \frac{4294}{25} & 0 & 0 & \frac{16299}{100}\\
        \end{pmatrix},\\
        N_{X_1Y_1X_2Y_2} =
        \begin{pmatrix}
            200 & 0 & 0 & \frac{952}{5} & 0 & 0 & 0 & 0 & 0 & 0 & 0 & 0 & \frac{1903}{10} & 0 & 0 & \frac{361}{2}\\
            0 & 0 & 0 & 0 & 0 & 0 & 0 & 0 & 0 & 0 & 0 & 0 & 0 & 0 & 0 & 0\\
            0 & 0 & \frac{93}{5} & 0 & 0 & 0 & 0 & 0 & 0 & 0 & 0 & 0 & 0 & 0 & \frac{37189}{2000} & 0\\
            \frac{952}{5} & 0 & 0 & \frac{3627}{20} & 0 & 0 & 0 & 0 & 0 & 0 & 0 & 0 & \frac{4528}{25} & 0 & 0 & \frac{4294}{25}\\
            0 & 0 & 0 & 0 & 0 & 0 & 0 & 0 & 0 & 0 & 0 & 0 & 0 & 0 & 0 & 0\\
            0 & 0 & 0 & 0 & 0 & 0 & 0 & 0 & 0 & 0 & 0 & 0 & 0 & 0 & 0 & 0\\
            0 & 0 & 0 & 0 & 0 & 0 & \frac{1}{20} & 0 & 0 & 0 & 0 & 0 & 0 & 0 & 0 & 0\\
            0 & 0 & 0 & 0 & 0 & 0 & 0 & 0 & 0 & 0 & 0 & 0 & 0 & 0 & 0 & 0\\
            0 & 0 & 0 & 0 & 0 & 0 & 0 & 0 & \frac{189}{10} & 0 & 0 & \frac{931}{50} & 0 & 0 & 0 & 0\\
            0 & 0 & 0 & 0 & 0 & 0 & 0 & 0 & 0 & 0 & 0 & 0 & 0 & 0 & 0 & 0\\
            0 & 0 & 0 & 0 & 0 & 0 & 0 & 0 & 0 & 0 & \frac{3}{50} & 0 & 0 & 0 & 0 & 0\\
            0 & 0 & 0 & 0 & 0 & 0 & 0 & 0 & \frac{931}{50} & 0 & 0 & \frac{459}{25} & 0 & 0 & 0 & 0\\
            \frac{1903}{10} & 0 & 0 & \frac{4528}{25} & 0 & 0 & 0 & 0 & 0 & 0 & 0 & 0 & \frac{1811}{10} & 0 & 0 & \frac{8589}{50}\\
            0 & 0 & 0 & 0 & 0 & 0 & 0 & 0 & 0 & 0 & 0 & 0 & 0 & 0 & 0 & 0\\
            0 & 0 & \frac{37189}{2000} & 0 & 0 & 0 & 0 & 0 & 0 & 0 & 0 & 0 & 0 & 0 & \frac{1859}{100} & 0\\
            \frac{361}{2} & 0 & 0 & \frac{4294}{25} & 0 & 0 & 0 & 0 & 0 & 0 & 0 & 0 & \frac{8589}{50} & 0 & 0 & \frac{16299}{100}\\
        \end{pmatrix}
    \end{gathered}
    \end{align}
    form a feasible solution to SDP~\eqref{eq:classical_capacity_def_sdp_dual} with $\tr\lrb{K_Y} = \frac{39}{20}$, implying that $C_0^{\bf FreeDef} (\,\vec{\cC}\,) \leq \log_2 \left\lfloor \frac{39}{20} \right\rfloor = \log_2 1 = 0$.
    Since, by definition, $C_0^{\bf FreeDef} (\,\vec{\cC}\,) \geq 0$, we conclude that $C_0^{\bf FreeDef} (\,\vec{\cC}\,) = 0$.
\end{proof}

\begin{figure}
    \centering
    \includegraphics[width=0.5\columnwidth]{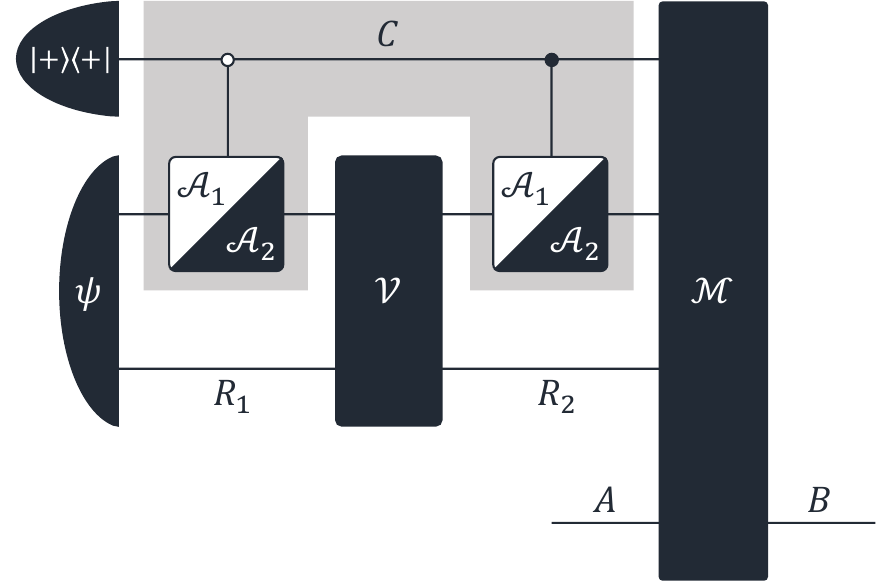}
    \caption{\textbf{Circuit implementation of the given QC-QC operation that achieves the advantage.}
    System $C$ is a control qubit that coherently controls the order of the two input channels $\cA_1$ and $\cA_2$.
    When the state of the control qubit is $\proj{0}$, $\cA_1$ is placed before $\cA_2$; when the state is $\proj{1}$, the order is reversed.
    Here, the state of the control qubit is set to $\proj{+}$.
    Auxiliary systems $R_1$ and $R_2$ have dimensions $2$ and $6$, respectively.
    The overall QC-QC operation can then be seen as the sequential application of two supermaps: the quantum switch (highlighted in gray) and a quantum supermap with definite order (the rest part of the quantum circuit).
    In this decomposition, the quantum switch is responsible for the indefinite causal order of the overall QC-QC operation, while the supermap with definite order makes the overall QC-QC operation signaling-non-generating.}
    \label{fig:QC_QC_implementation}
\end{figure}

In the proof above, we give the Choi operator of a QC-QC operation that achieves the communication of one classical bit.
Fig.~\ref{fig:QC_QC_implementation} shows an explicit circuit implementation of this QC-QC operation, which is obtained from the procedure given in Appendix~B.3.c of Ref.~\cite{wechs2021quantum}.
The control qubit is fixed to the state $\proj{+}$, and the initial state $\psi$ is fixed to $\psi \coloneqq \proj{\psi}$ with $\ket{\psi} \coloneqq \sqrt{\frac{13}{40}}\ket{00} + \sqrt{\frac{27}{40}}\ket{11}$.
The channel $\cV$ is defined as $\cV(\cdot) \coloneqq V(\cdot)V^\dagger$, where
\begin{align}
    V = \proj{00} + \proj{11} + \lrp{\sqrt{\frac{22}{27}}\ket{02} + \sqrt{\frac{5}{27}}\ket{13}}\bra{01} + \lrp{\sqrt{\frac{11}{13}}\ket{04} + \sqrt{\frac{2}{13}}\ket{15}}\bra{10}
\end{align}
is an isometry with input systems $Y_1$ (or $Y_2$) and $R_1$ and output systems $X_2$ (or $X_1$) and $R_2$.
For the quantum channel $\cM$, we give its Choi operator as follows:
\begin{align}
    J^\cM_{CTR_2AB} &\coloneqq \frac{1}{2} \idop_{{CTR_2AB}} + \frac{1}{2} \left[
    \idop_C \otimes \left( \ketbra{0}{0} \otimes
        \begin{pmatrix}
            -\frac{3}{26} & 0 & 0 & 0 & 0 & 0 \\
            0 & \frac{9}{64} & 0 & 0 & 0 & 0 \\
            0 & 0 & -\frac{1}{11} & 0 & 0 & 0 \\
            0 & 0 & 0 & \frac{1}{2} & 0 & 0 \\
            0 & 0 & 0 & 0 & -\frac{11}{24} & 0 \\
            0 & 0 & 0 & 0 & 0 & \frac{7}{108} \\
        \end{pmatrix}
        + \ketbra{0}{1} \otimes
        \begin{pmatrix}
            0 & \frac{80}{27 \sqrt{39}} & 0 & 0 & 0 & 0 \\
            0 & 0 & 0 & 0 & 0 & 0 \\
            0 & 0 & 0 & \frac{5}{4 \sqrt{11}} & 0 & 0 \\
            0 & 0 & 0 & 0 & 0 & 0 \\
            0 & 0 & 0 & 0 & 0 & 0 \\
            0 & 0 & 0 & 0 & 0 & 0 \\
        \end{pmatrix} \right. \right. \notag \\
    & \left. \hspace{40mm} + \ketbra{1}{0} \otimes
        \begin{pmatrix}
            0 & 0 & 0 & 0 & 0 & 0 \\
            \frac{80}{27 \sqrt{39}} & 0 & 0 & 0 & 0 & 0 \\
            0 & 0 & 0 & 0 & 0 & 0 \\
            0 & 0 & \frac{5}{4 \sqrt{11}} & 0 & 0 & 0 \\
            0 & 0 & 0 & 0 & 0 & 0 \\
            0 & 0 & 0 & 0 & 0 & 0 \\
        \end{pmatrix}
        + \ketbra{1}{1} \otimes
        \begin{pmatrix}
            -\frac{13}{24} & 0 & 0 & 0 & 0 & 0 \\
            0 & \frac{587}{4374} & 0 & 0 & 0 & 0 \\
            0 & 0 & \frac{77}{108} & 0 & 0 & 0 \\
            0 & 0 & 0 & \frac{5}{192} & 0 & 0 \\
            0 & 0 & 0 & 0 & -\frac{2207}{5346} & 0 \\
            0 & 0 & 0 & 0 & 0 & \frac{223}{243} \\
        \end{pmatrix}\right) \notag \\
    & \hspace{12mm} + \sigma_1 \otimes \left( \ketbra{0}{0} \otimes \begin{pmatrix}
            -\frac{3}{26} & 0 & 0 & 0 & 0 & 0 \\
             0 & 0 & \frac{3}{4} \sqrt{\frac{15}{11}} & 0 & 0 & 0 \\
             0 & \frac{3}{4} \sqrt{\frac{15}{11}} & 0 & 0 & 0 & 0 \\
             0 & 0 & 0 & \frac{1}{2} & 0 & 0 \\
             0 & 0 & 0 & 0 & 0 & 0 \\
             0 & 0 & 0 & 0 & 0 & 0 \\
        \end{pmatrix}
        + \ketbra{0}{1} \otimes
        \begin{pmatrix}
            0 & \frac{80}{27 \sqrt{39}} & 0 & 0 & 0 & 0 \\
            0 & 0 & 0 & \frac{\sqrt{15}}{64} & 0 & 0 \\
            0 & 0 & 0 & 0 & 0 & 0 \\
            0 & 0 & 0 & 0 & 0 & 0 \\
            -\frac{\sqrt{143}}{24} & 0 & 0 & 0 & 0 & 0 \\
            0 & 0 & \frac{7 \sqrt{11}}{108} & 0 & 0 & 0 \\
        \end{pmatrix} \right. \notag \\
    & \left. \left. \hspace{24mm} + \ketbra{1}{0} \otimes
        \begin{pmatrix}
            0 & 0 & 0 & 0 & -\frac{\sqrt{143}}{24} & 0 \\
            \frac{80}{27 \sqrt{39}} & 0 & 0 & 0 & 0 & 0 \\
            0 & 0 & 0 & 0 & 0 & \frac{7 \sqrt{11}}{108} \\
            0 & \frac{\sqrt{15}}{64} & 0 & 0 & 0 & 0 \\
            0 & 0 & 0 & 0 & 0 & 0 \\
            0 & 0 & 0 & 0 & 0 & 0 \\
        \end{pmatrix}
        + \ketbra{1}{1} \otimes
        \begin{pmatrix}
            0 & 0 & 0 & 0 & 0 & 0 \\
            0 & \frac{587}{4374} & 0 & 0 & 0 & 0 \\
            0 & 0 & 0 & 0 & 0 & 0 \\
            0 & 0 & 0 & 0 & 0 & 0 \\
            0 & 0 & 0 & 0 & -\frac{2207}{5346} & 0 \\
            0 & 0 & 0 & 0 & 0 & 0 \\
        \end{pmatrix} \right)
    \right] \otimes \sigma_3 \otimes \sigma_3,
\end{align}
where the input system $T$ stores the state output by the second channel in the circuit.
One can check that the circuit in Fig.~~\ref{fig:QC_QC_implementation} indeed reproduces the desired supermap by verifying that the Choi operator $J_{XYAB}$ in Eq.~\eqref{supp_eq:choi_from_feasible}, with $E_{XY}$ and $F_{XY}$ specified in Eq.~\eqref{supp_eq:gen_primal_solution}, is equal to $J_{XYCTR_2}^{\rm QS} \star J^\cM_{CTR_2AB}$, where $J_{XYCTR_2}^{\rm QS} \coloneqq \dproj{J^{\rm QS}}$ is the Choi operator (process matrix) of the ``quantum-switch-like'' part of the circuit (the upper-left part of the circuit that is fed into $\cM$) with
\begin{align}
    \dket{J^{\rm QS}} \coloneqq \frac{1}{\sqrt{2}} \ket{0}_C \ox \ket{\psi}_{X_1R_1} \star \dket{V}_{Y_1R_1X_2R_2} \ox \ket{\Gamma}_{Y_2T} + \frac{1}{\sqrt{2}} \ket{1}_C \ox \ket{\psi}_{X_2R_1} \star \dket{V}_{Y_2R_1X_1R_2} \ox \ket{\Gamma}_{Y_1T},
\end{align}
where $\ket{\Gamma}_{Y_2T} \coloneqq \sum_j \ket{j}_{Y_2} \ox \ket{j}_T$ is the unnormalized pure maximally entangled state.
The ``double ket'' notation $\dket{V}_{Y_1R_1X_2R_2} \coloneqq (I_{Y_1R_1} \ox V) \ket{\Gamma}_{Y_1R_1 Y_1R_1}$ denotes the pure Choi vector of $V$.
The link product of two vectors is defined as $\ket{p}_{AB} \star \ket{q}_{BC} \coloneqq (I_{AC} \ox \bra{\Gamma}_{BB})(\ket{p}_{AB} \ox \ket{q}_{BC})$, and thus $\ket{\psi}_{X_1R_1} \star \dket{V}_{Y_1R_1X_2R_2} = (I_{X_1Y_1X_2R_2} \ox \bra{\Gamma}_{R_1R_1})(\ket{\psi}_{X_1R_1} \ox \dket{V}_{Y_1R_1X_2R_2})$.

In the proof of Proposition~\ref{supp_prop:diamond_prob_equal}, we restrict the search for optimal supermaps to those whose Choi operators $J^\mS_{XYA_mB_m}$ are of the form $J^\mS_{XYA_mB_m} = E_{XY}\ox J^{\Delta_m}_{A_mB_m} + F_{XY}\ox\lrp{\idop_{A_mB_m}-J^{\Delta_m}_{A_mB_m}}$.
Notice that any operator $J^\mS_{XYA_mB_m}$ of this form satisfies $J^\mS_{XYA_m} = (E_{XY}+(m-1)F_{XY})\ox\idop_{A_m}$, which means that $A_m$ does not signal to $XY$, and thus the system $A$ we want to communicate to $B$ can enter the circuit after the use of the original channels.
Thus, what we obtain in our specific QC-QC implementation above is a general feature.

%%%%%%%%%%%%%%%%%%%%%%%%%%%%%%%%%%%%%%%%%%%%%%%%%%%%%%%%%%%%
%%%%%%%%%%%%%%%%%%%%%%%%%%%%%%%%%%%%%%%%%%%%%%%%%%%%%%%%%%%%
\section{No advantage for Pauli channels}\label{supp_sec:pauli}
\begin{theorem}
    Given a list of $N$, possibly different, Pauli channels $\vec{\cP}$, it holds that
    \begin{align}
        C_\epsilon^{\bf Free} (\,\vec{\cP}\,) = C_\epsilon^{\bf FreePar} (\,\vec{\cP}\,)
    \end{align}
    for every error tolerance $\epsilon \geq 0$.
\end{theorem}
\begin{proof}
    For a $q$-qubit Pauli channel $\cP(\cdot) = \sum_{i_1\dots i_q} p_{i_1\dots i_q} (\sigma_{i_1} \ox \cdots \ox \sigma_{i_q}) (\cdot) (\sigma_{i_1} \ox \cdots \ox \sigma_{i_q})$, its Choi operator is
    \begin{align}
        J^\cP_{XY} = \sum_{i_1\dots i_q} p_{i_1\dots i_q} \lrp{\idop_X \ox (\sigma_{i_1} \ox \cdots \ox \sigma_{i_q})_Y} \proj{\Gamma}_{XY} \lrp{\idop_X \ox (\sigma_{i_1} \ox \cdots \ox \sigma_{i_q})_Y} = \sum_{i_1\dots i_q} p_{i_1\dots i_q} \Gamma^{i_1\dots i_q}_{XY},
    \end{align}
    where $\Gamma^{i_1\dots i_q}_{XY} \coloneqq \lrp{\idop_X \ox (\sigma_{i_1} \ox \cdots \ox \sigma_{i_q})_Y} \proj{\Gamma}_{XY} \lrp{\idop_X \ox (\sigma_{i_1} \ox \cdots \ox \sigma_{i_q})_Y}$ and $\ket{\Gamma}_{XY} \coloneqq \sum_{j=1}^{2^q} \ket{j}_X\ox \ket{j}_Y$.
    The Choi operator of $N$ Pauli channels $\vec{\cP} = (\cP_1, \dots, \cP_N)$ with the $j$-th Pauli channel $\cP_j$ acting on $q_j$ number of qubits is
    \begin{align}
        J^{\vec{\cP}}_{XY} = \bigotimes_{j=1}^N\lrp{\sum_{i_1\dots i_{q_j}} p_{i_1\dots i_{q_j}} \Gamma^{i_1\dots i_{q_j}}_{X_jY_j}} = \sum_{i_1\dots i_{q_{\rm tot}}} p_{i_1\dots i_{q_{\rm tot}}} \Gamma^{i_1\dots i_{q_{\rm tot}}}_{XY},
    \end{align}
    where $q_{\rm tot} = \sum_{j=1}^N q_j$ is the total number of qubits $\vec{\cP}$ acts on, and each $p_{i_1\dots i_{q_{\rm tot}}}$ is a probability determined by all $p_{i_1\dots i_{q_j}}$.
    Then, according to Proposition~\ref{prop:gen_capacity_sdp}, the one-shot $\epsilon$-error classical capacity of $\vec{\cP}$ assisted by general supermaps can be formulated as
    \begin{align}
    \begin{aligned}\label{eq:sdp_pauli_capacity_gen}
        C_\epsilon^{\bf Free} (\,\vec{\cP}\,) = \log_2 \max &\; \lfloor m \rfloor\\
        \text{\rm s.t.} &\; E_{XY} \star \sum_{i_1\dots i_{q_{\rm tot}}} p_{i_1\dots i_{q_{\rm tot}}} \Gamma^{i_1\dots i_{q_{\rm tot}}}_{XY} \geq m(1-\epsilon),\, E_Y = \idop_Y,\\
        &\; 0 \leq E_{XY} \leq F_{XY},\, \cL^{\rm NS}_{XY}\lrp{F_{XY}} = \frac{m \idop_{XY}}{2^{q_{\rm tot}}}.
    \end{aligned}
    \end{align}
    Assume that $m^*$, $E^*_{XY}$, and $F^*_{XY}$ form an optimal solution to the above maximization program.
    For every $i_1\dots i_{q_{\rm tot}}$, define $\alpha_{i_1\dots i_{q_{\rm tot}}} \coloneqq \frac{1}{4^{q_{\rm tot}}} \tr\lrb{E^*_{XY} \Gamma^{i_1\dots i_{q_{\rm tot}}}_{XY}}$ and $\beta_{i_1\dots i_{q_{\rm tot}}} \coloneqq \frac{1}{4^{q_{\rm tot}}} \tr\lrb{F^*_{XY} \Gamma^{i_1\dots i_{q_{\rm tot}}}_{XY}}$.
    Then, $m^*$, $\widetilde{E}_{XY} \coloneqq \sum_{i_1\dots i_{q_{\rm tot}}} \alpha_{i_1\dots i_{q_{\rm tot}}} \Gamma^{i_1\dots i_{q_{\rm tot}}}_{XY}$, and $\widetilde{F}_{XY} \coloneqq \sum_{i_1\dots i_{q_{\rm tot}}} \beta_{i_1\dots i_{q_{\rm tot}}} \Gamma^{i_1\dots i_{q_{\rm tot}}}_{XY}$ also form an optimal solution, as we shall show now.

    For the error tolerance constraint, observe that
    \begin{align}
        \widetilde{E}_{XY} \star \sum_{i_1\dots i_{q_{\rm tot}}} p_{i_1\dots i_{q_{\rm tot}}} \Gamma^{i_1\dots i_{q_{\rm tot}}}_{XY} &= \sum_{i_1\dots i_{q_{\rm tot}}} p_{i_1\dots i_{q_{\rm tot}}} \widetilde{E}_{XY} \star \Gamma^{i_1\dots i_{q_{\rm tot}}}_{XY}\\
        &= \sum_{i_1\dots i_{q_{\rm tot}}} \sum_{j_1\dots j_{q_{\rm tot}}} \alpha_{j_1\dots j_{q_{\rm tot}}} p_{i_1\dots i_{q_{\rm tot}}} \Gamma^{j_1\dots j_{q_{\rm tot}}}_{XY} \star \Gamma^{i_1\dots i_{q_{\rm tot}}}_{XY}\\
        &= \sum_{i_1\dots i_{q_{\rm tot}}} \alpha_{i_1\dots i_{q_{\rm tot}}} p_{i_1\dots i_{q_{\rm tot}}} \cdot 4^{q_{\rm tot}}\\
        &= \sum_{i_1\dots i_{q_{\rm tot}}} p_{i_1\dots i_{q_{\rm tot}}} \tr\lrb{E^*_{XY} \Gamma^{i_1\dots i_{q_{\rm tot}}}_{XY}}\\
        &= E^*_{XY} \star \sum_{i_1\dots i_{q_{\rm tot}}} p_{i_1\dots i_{q_{\rm tot}}} \Gamma^{i_1\dots i_{q_{\rm tot}}}_{XY}\\
        &\geq m^*(1-\epsilon),
    \end{align}
    where in the third equality, we use the identity $\tr\lrb{\Gamma^{j_1\dots j_{q_{\rm tot}}}_{XY} \Gamma^{i_1\dots i_{q_{\rm tot}}}_{XY}} = 4^{q_{\rm tot}} \delta_{i_1\dots i_{q_{\rm tot}},j_1\dots j_{q_{\rm tot}}}$.
    Since $E^*_{XY} \geq 0$, we know $\tr\lrb{E^*_{XY} \Gamma^{i_1\dots i_{q_{\rm tot}}}_{XY}}$ and thus $\alpha_{i_1\dots i_{q_{\rm tot}}}$ are nonnegative for every $i_1\dots i_{q_{\rm tot}}$, implying that $\widetilde{E}_{XY} \geq 0$.
    Similarly, because $E^*_{XY} \leq F^*_{XY}$, we have $\tr\lrb{E^*_{XY} \Gamma^{i_1\dots i_{q_{\rm tot}}}_{XY}} \leq \tr\lrb{F^*_{XY} \Gamma^{i_1\dots i_{q_{\rm tot}}}_{XY}}$, and thus $\alpha_{i_1\dots i_{q_{\rm tot}}} \leq \beta_{i_1\dots i_{q_{\rm tot}}}$ for every $i_1\dots i_{q_{\rm tot}}$, implying that $\widetilde{E}_{XY} \leq \widetilde{F}_{XY}$.
    For the causality constraint, first note that, for every $i_1\dots i_{q_{\rm tot}}$,
    \begin{align}
        \tr\lrb{F^*_{XY} \Gamma^{i_1\dots i_{q_{\rm tot}}}_{XY}} &= \tr\lrb{F^*_{XY} \cL^{\rm NS}_{XY}\lrp{\Gamma^{i_1\dots i_{q_{\rm tot}}}_{XY}}}\\
        &= \tr\lrb{\cL^{\rm NS}_{XY}\lrp{F^*_{XY}} \Gamma^{i_1\dots i_{q_{\rm tot}}}_{XY}}\\
        &= \frac{m^*}{2^{q_{\rm tot}}} \tr\lrb{\idop_{XY} \Gamma^{i_1\dots i_{q_{\rm tot}}}_{XY}}\\
        &= m^*,
    \end{align}
    where the first equality is due to that $\Gamma^{i_1\dots i_{q_{\rm tot}}}_{XY}$ is the Choi operator of a no-signaling $N$-partite channel, and the second equality holds because the projection $\cL^{\rm NS}_{XY}$ is self-dual.
    Then, $\beta_{i_1\dots i_{q_{\rm tot}}} = \frac{m^*}{4^{q_{\rm tot}}}$ for every $i_1\dots i_{q_{\rm tot}}$, resulting in
    \begin{align}\label{eq:widetilde_f}
        \widetilde{F}_{XY} = \frac{m^*}{4^{q_{\rm tot}}} \sum_{i_1\dots i_{q_{\rm tot}}} \Gamma^{i_1\dots i_{q_{\rm tot}}}_{XY} = \frac{m^* \idop_{XY}}{2^{q_{\rm tot}}},
    \end{align}
    where the second equality is due to the identity $\sum_{i_1\dots i_{q_{\rm tot}}} \Gamma^{i_1\dots i_{q_{\rm tot}}}_{XY} = 2^{q_{\rm tot}} \idop_{XY}$.
    By Eq.~\eqref{eq:widetilde_f}, it is straightforward to see that the causality constraint $\cL^{\rm NS}_{XY}\lrp{\widetilde{F}_{XY}} = \frac{m^* \idop_{XY}}{2^{q_{\rm tot}}}$ holds.
    Finally, consider that
    \begin{align}
        \widetilde{E}_Y &= \sum_{i_1\dots i_{q_{\rm tot}}} \alpha_{i_1\dots i_{q_{\rm tot}}} \tr_X\lrb{\Gamma^{i_1\dots i_{q_{\rm tot}}}_{XY}}\\
        &= \frac{1}{4^{q_{\rm tot}}} \sum_{i_1\dots i_{q_{\rm tot}}} \tr\lrb{E^*_{XY} \Gamma^{i_1\dots i_{q_{\rm tot}}}_{XY}} \idop_Y\\
        &= \frac{1}{4^{q_{\rm tot}}} \tr\lrb{E^*_{XY} \lrp{2^{q_{\rm tot}} \idop_{XY}}} \idop_Y\\
        &= \frac{\tr\lrb{E^*_{XY}}}{2^{q_{\rm tot}}} \idop_Y\\
        &= \idop_Y,
    \end{align}
    where the last equality is due to $\tr\lrb{E^*_{XY}} = \tr\lrb{E^*_Y} = \tr\lrb{\idop_Y} = 2^{q_{\rm tot}}$.

    So far, we have shown that $m^*$, $\widetilde{E}_{XY}$, and $\widetilde{F}_{XY}$ satisfy all the constraints in SDP~\eqref{eq:sdp_pauli_capacity_gen}.
    Since $m^*$ is optimal, they form an optimal solution.
    Now, by Eq.~\eqref{appeq:par_capacity_sdp},
    \begin{align}
    \begin{aligned}
        C_\epsilon^{\bf FreePar} (\,\vec{\cP}\,) = \log_2 \max &\; \lfloor m \rfloor\\
        \text{\rm s.t.} &\; E_{XY} \star \sum_{i_1\dots i_{q_{\rm tot}}} p_{i_1\dots i_{q_{\rm tot}}} \Gamma^{i_1\dots i_{q_{\rm tot}}}_{XY} \geq m(1-\epsilon),\, E_Y = \idop_Y,\\
        &\; 0 \leq E_{XY} \leq F_X \ox \idop_Y,\, \tr\lrb{F_X} = m.
    \end{aligned}
    \end{align}
    Clearly, $m^*$, $\widetilde{E}_{XY}$, and $\frac{1}{2^{q_{\rm tot}}}\widetilde{F}_X$ form a feasible solution as $\widetilde{F}_{XY} = \frac{m^*}{2^{q_{\rm tot}}} \idop_X \ox \idop_Y = \frac{1}{2^{q_{\rm tot}}}\widetilde{F}_X \ox \idop_Y$.
    Thus, we know $C_\epsilon^{\bf FreePar} (\,\vec{\cP}\,) \geq \log_2 \lfloor m^* \rfloor$.
    Since $C_\epsilon^{\bf FreePar} (\,\vec{\cP}\,) \leq C_\epsilon^{\bf Free} (\,\vec{\cP}\,)$ and $C_\epsilon^{\bf Free} (\,\vec{\cP}\,) = \log_2 \lfloor m^* \rfloor$, we conclude that $C_\epsilon^{\bf FreePar} (\,\vec{\cP}\,) = C_\epsilon^{\bf Free} (\,\vec{\cP}\,)$.
\end{proof}

%%%%%%%%%%%%%%%%%%%%%%%%%%%%%%%%%%%%%%%%%%%%%%%%%%%%%%%%%%%%
%%%%%%%%%%%%%%%%%%%%%%%%%%%%%%%%%%%%%%%%%%%%%%%%%%%%%%%%%%%%
\section{No advantage for channel simulation}\label{supp_sec:chan_sim}
\begin{theorem}
    For any quantum channel $\cC$ and error tolerance $\epsilon \geq 0$, the one-shot classical simulation cost assisted by free general supermaps is equal to that assisted by free parallel supermaps, i.e.,
    \begin{align}
        S_\epsilon^{\bf Free}(\cC) = S_\epsilon^{\bf FreePar}(\cC).
    \end{align}
\end{theorem}
\begin{proof}
    By definition, it is straightforward to see that $S_\epsilon^{\bf Free}(\cC) \leq S_\epsilon^{\bf FreePar}(\cC)$ since ${\bf FreePar} \subset {\bf Free}$.
    Now, we show that $S_\epsilon^{\bf FreePar}(\cC)$ is also a lower bound to $S_\epsilon^{\bf Free}(\cC)$.

    The first step is to express $S_\epsilon^{\bf Free}(\cC)$ as an optimization program in terms of the Choi operator of channel $\cC$ with input $A$ and output $B$.
    The diamond distance between two quantum channels $\cC_1$ and $\cC_2$ can be written as an SDP~\cite{watrous2009semidefinite}:
    \begin{align}
    \begin{aligned}
        \frac{1}{2}\lrV{\cC_1 - \cC_2}_\diamond = \min &\; \mu\\
        \text{\rm s.t.} &\; Z_{AB} \geq 0,\, Z_{AB} \geq J^{\cC_1}_{AB} - J^{\cC_2}_{AB},\, \tr_B\lrb{Z_{AB}} \leq \mu\idop_A,
    \end{aligned}
    \end{align}
    where $J^{\cC_1}_{AB}$ and $J^{\cC_2}_{AB}$ are the Choi operators of the corresponding channels.
    It then follows that the simulation cost $S_\epsilon^{\bf Free}(\cC)$ can be formulated as
    \begin{align}
    \begin{aligned}
        S_\epsilon^{\bf Free}(\cC) = \log_2 \min &\; m\\
        \text{\rm s.t.} &\; Z_{AB} \geq 0,\, Z_{AB} \geq J^\cM_{AB} - J^\cC_{AB},\, Z_A \leq \epsilon\idop_A,\, J^\cM_{AB} = J^\mS_{XYAB} \star J^{\vec{\Delta}_m}_{XY},\\
        &\; J^\mS_{XYAB} \ge 0,\, \cL^{\rm NS}_{XY}\lrp{J^\mS_{XYA}} = \frac{\idop_{XYA}}{d_X},\, {}_{[1-A]}J^\mS_{YAB} = 0,
    \end{aligned}
    \end{align}
    where $m$ is minimized over positive integers.
    The constraints in the last line is requiring $\mS$ to be a supermap in ${\bf Free}$.
    We obtain a lower bound to $S_\epsilon^{\bf Free}(\cC)$ by relaxing the constraint $\cL^{\rm NS}_{XY}\lrp{J^\mS_{XYA}} = \frac{\idop_{XYA}}{d_X}$ into constraints $J^\cM_A = \idop_A$ and $J^\mS_{YA} = \idop_{YA}$.
    To show that this is indeed a relaxation, we now prove that these two relaxed constraints are implications of $\cL^{\rm NS}_{XY}\lrp{J^\mS_{XYA}} = \frac{\idop_{XYA}}{d_X}$.
    For the first constraint, consider that
    \begin{align}
        J^\cM_A &= \tr_B\lrb{J^\mS_{XYAB} \star J^{\vec{\Delta}_m}_{XY}}\\
        &= \tr_{XYB}\lrb{J^\mS_{XYAB} \lrp{\cL^{\rm NS}_{XY}\lrp{J^{\vec{\Delta}_m}_{XY}} \ox \idop_{AB}}}\\
        &= \tr_{XY}\lrb{\cL^{\rm NS}_{XY}\lrp{J^\mS_{XYA}} \lrp{J^{\vec{\Delta}_m}_{XY} \ox \idop_A}}\\
        &= \tr_{XY}\lrb{\frac{\idop_{XYA}}{d_X} \lrp{J^{\vec{\Delta}_m}_{XY} \ox \idop_A}}\\
        &= \idop_A,
    \end{align}
    which is what we set out to show.
    For the other constraint, note that, in one direction, we have
    \begin{align}
        \tr_X\lrb{\cL^{\rm NS}_{XY}\lrp{J^\mS_{XYA}}} &= \tr_X\lrb{\frac{\idop_{XYA}}{d_X}}\\
        &= \idop_{YA}.\label{appeq:one_direction}
    \end{align}
    In another direction, due to the observation that $\tr_X[\cL^{\rm NS}_{XY}(\cdot)] = \tr_X[\cdot]$, we have
    \begin{align}\label{supp_eq:causality_partial_trace}
        \tr_X\lrb{\cL^{\rm NS}_{XY}\lrp{J^\mS_{XYA}}} = J^\mS_{YA},
    \end{align}
    which, together with Eq.~\eqref{appeq:one_direction}, implies
    \begin{align}
        J^\mS_{YA} = \idop_{YA}.
    \end{align}
    Hence, we arrive at
    \begin{align}
    \begin{aligned}
        S_\epsilon^{\bf Free}(\cC) \geq \log_2 \min &\; m\\
        \text{\rm s.t.} &\; Z_{AB} \geq 0,\, Z_{AB} \geq J^\cM_{AB} - J^\cC_{AB},\, Z_A \leq \epsilon\idop_A,\, J^\cM_{AB} = J^\mS_{XYAB} \star J^{\vec{\Delta}_m}_{XY},\\
        &\; J^\mS_{XYAB} \ge 0,\, J^\cM_A = \idop_A,\, J^\mS_{YA} = \idop_{YA},\, {}_{[1-A]}J^\mS_{YAB} = 0.
    \end{aligned}
    \end{align}

    Similar to what we did in deriving the SDPs for the capacities, here we use the symmetry of $J^{\vec{\Delta}_m}_{XY}$ to simplify the above optimization program, arriving at
    \begin{align}
    \begin{aligned}
        S_\epsilon^{\bf Free}(\cC) \geq \log_2 \min &\; \lceil\tr\lrb{F_B}\rceil\\
        \text{\rm s.t.} &\; Z_{AB} \geq 0,\, Z_{AB} \geq E_{AB} - J^\cC_{AB},\, Z_A \leq \epsilon\idop_A,\\
        &\; 0 \leq E_{AB} \leq \idop_A \ox F_B,\, E_A = \idop_A.
    \end{aligned}
    \end{align}
    Observe that the right hand side is exactly the SDP for one-shot classical simulation cost assisted by parallel supermaps as given in Ref.~\cite{fang2019quantum}.
    Hence, $S_\epsilon^{\bf Free}(\cC) \geq S_\epsilon^{\bf FreePar}(\cC)$.
    Therefore, we conclude that $S_\epsilon^{\bf Free}(\cC) = S_\epsilon^{\bf FreePar}(\cC)$.
\end{proof}
\end{document}

%% file: pretex.tex
\usepackage{algorithm,algorithmic,amsfonts,amsmath,amssymb,mathtools,mathrsfs,bbm,float, setspace}
\usepackage{graphicx,epic,eepic,epsfig,latexsym,verbatim,color}
\usepackage[shortlabels]{enumitem}

\usepackage{tikz}
\usetikzlibrary{chains}
\usetikzlibrary{fit}

\usepackage{epsfig}
\usetikzlibrary{shapes.symbols,patterns} % for source symbols
\usepackage{pgfplots}
\pgfplotsset{compat=1.18}

\usepackage[strict]{changepage}
\usepackage{hyperref}
\usepackage{hyperref}
\hypersetup{colorlinks=true,citecolor=blue,linkcolor=blue,filecolor=blue,urlcolor=blue,breaklinks=true}

\usepackage[marginal]{footmisc}
\usepackage{url}
\usepackage{theorem}

\usepackage{framed}
\definecolor{shadecolor}{rgb}{0.9,0.9,0.9}

\def\squareforqed{\hbox{\rlap{$\sqcap$}$\sqcup$}}
\def\qed{\ifmmode\squareforqed\else{\unskip\nobreak\hfil
\penalty50\hskip1em\null\nobreak\hfil\squareforqed
\parfillskip=0pt\finalhyphendemerits=0\endgraf}\fi}
\def\endenv{\ifmmode\;\else{\unskip\nobreak\hfil
\penalty50\hskip1em\null\nobreak\hfil\;
\parfillskip=0pt\finalhyphendemerits=0\endgraf}\fi}
\newenvironment{proof}{\noindent \textbf{{Proof~} }}{\hfill $\blacksquare$}

\mathchardef\ordinarycolon\mathcode`\:
\mathcode`\:=\string"8000
\def\vcentcolon{\mathrel{\mathop\ordinarycolon}}
\begingroup \catcode`\:=\active
  \lowercase{\endgroup
  \let :\vcentcolon
  }

\usepackage{cleveref}
\usepackage{graphicx}
\usepackage{xcolor}

\definecolor{darkblue}{RGB}{0,76,156}
\definecolor{darkkblue}{RGB}{0,0,153}
\definecolor{blue2}{RGB}{102,178,255}
\definecolor{darkred}{RGB}{195,0,0}

\RequirePackage[framemethod=default]{mdframed}
\newmdenv[skipabove=7pt,
skipbelow=7pt,
backgroundcolor=darkblue!15,
innerleftmargin=5pt,
innerrightmargin=5pt,
innertopmargin=5pt,
leftmargin=0cm,
rightmargin=0cm,
innerbottommargin=5pt,
linewidth=1pt]{tBox}

\newmdenv[skipabove=7pt,
skipbelow=7pt,
backgroundcolor=blue2!25,
innerleftmargin=5pt,
innerrightmargin=5pt,
innertopmargin=5pt,
leftmargin=0cm,
rightmargin=0cm,
innerbottommargin=5pt,
linewidth=1pt]{dBox}

\newmdenv[skipabove=7pt,
skipbelow=7pt,
backgroundcolor=darkred!15,
innerleftmargin=5pt,
innerrightmargin=5pt,
innertopmargin=5pt,
leftmargin=0cm,
rightmargin=0cm,
innerbottommargin=5pt,
linewidth=1pt]{rBox}

\newcommand{\nc}{\newcommand}
\nc{\bra}[1]{\langle#1|}
\nc{\ket}[1]{|#1\rangle}
\nc{\dket}[1]{|#1\rangle\!\rangle}
\nc{\ketbra}[2]{|#1\rangle\!\langle#2|}
\nc{\braket}[2]{\langle#1|#2\rangle}

\DeclarePairedDelimiter{\floor}{\lfloor}{\rfloor}

\nc{\proj}[1]{| #1\rangle\!\langle #1 |}
\nc{\dproj}[1]{| #1\rangle\!\rangle\!\langle\!\langle #1 |}
\nc{\avg}[1]{\langle#1\rangle}
\nc{\rank}{\operatorname{Rank}}
\nc{\smfrac}[2]{\mbox{$\frac{#1}{#2}$}}
\nc{\tr}{\operatorname{Tr}}
\nc{\ox}{\otimes}
\nc{\dg}{\dagger}
\nc{\dn}{\downarrow}
\nc{\cA}{{\cal A}}
\nc{\cB}{{\cal B}}
\nc{\cC}{{\cal C}}
\nc{\cD}{{\cal D}}
\nc{\cE}{{\cal E}}
\nc{\cF}{{\cal F}}
\nc{\cG}{{\cal G}}
\nc{\cH}{{\cal H}}
\nc{\cI}{{\cal I}}
\nc{\cJ}{{\cal J}}
\nc{\cK}{{\cal K}}
\nc{\cL}{{\cal L}}
\nc{\cM}{{\cal M}}
\nc{\cN}{{\cal N}}
\nc{\cO}{{\cal O}}
\nc{\cP}{{\cal P}}
\nc{\cQ}{{\cal Q}}
\nc{\cR}{{\cal R}}
\nc{\cS}{{\cal S}}
\nc{\cT}{{\cal T}}
\nc{\cU}{{\cal U}}
\nc{\cV}{{\cal V}}
\nc{\cX}{{\cal X}}
\nc{\cY}{{\cal Y}}
\nc{\cZ}{{\cal Z}}
\nc{\cW}{{\cal W}}
\nc{\csupp}{{\operatorname{csupp}}}
\nc{\qsupp}{{\operatorname{qsupp}}}
\nc{\var}{{\operatorname{var}}}
\nc{\rar}{\rightarrow}
\nc{\lrar}{\longrightarrow}
\nc{\polylog}{{\operatorname{polylog}}}
\nc{\wt}{{\operatorname{wt}}}
\nc{\supp}{{\operatorname{supp}}}

\def\i{\mathbf{i}}

\def\x{\xi}

\nc{\RR}{{{\mathbb R}}}
\nc{\CC}{{{\mathbb C}}}
\nc{\FF}{{{\mathbb F}}}
\nc{\NN}{{{\mathbb N}}}
\nc{\ZZ}{{{\mathbb Z}}}
\nc{\PP}{{{\mathbb P}}}
\nc{\QQ}{{{\mathbb Q}}}
\nc{\UU}{{{\mathbb U}}}
\nc{\EE}{{{\mathbb E}}}
\nc{\id}{{\operatorname{id}}}

\nc{\CHSH}{{\operatorname{CHSH}}}

\nc{\rU}{\mbox{U}}

\nc{\ob}[1]{#1}

\nc{\SEP}{{\text{\rm SEP}}}
\nc{\NS}{{\text{\rm NS}}}
\nc{\LOCC}{{\text{\rm LOCC}}}
\nc{\PPT}{{\text{\rm PPT}}}
\nc{\EXT}{{\text{\rm EXT}}}
\nc{\Sym}{{\operatorname{Sym}}}

\nc{\ERLO}{{E_{\text{r,LO}}}}
\nc{\ERLOCC}{{E_{\text{r,LOCC}}}}
\nc{\ERPPT}{{E_{\text{r,PPT}}}}
\nc{\ERLOCCinfty}{{E^{\infty}_{\text{r,LOCC}}}}
\nc{\Aram}{{\operatorname{\sf A}}}
\makeatletter
\def\grd@save@target#1{%
  \def\grd@target{#1}}
\def\grd@save@start#1{%
  \def\grd@start{#1}}
\tikzset{
  grid with coordinates/.style={
    to path={%
      \pgfextra{%
        \edef\grd@@target{(\tikztotarget)}%
        \tikz@scan@one@point\grd@save@target\grd@@target\relax
        \edef\grd@@start{(\tikztostart)}%
        \tikz@scan@one@point\grd@save@start\grd@@start\relax
        \draw[minor help lines,magenta] (\tikztostart) grid (\tikztotarget);
        \draw[major help lines] (\tikztostart) grid (\tikztotarget);
        \grd@start
        \pgfmathsetmacro{\grd@xa}{\the\pgf@x/1cm}
        \pgfmathsetmacro{\grd@ya}{\the\pgf@y/1cm}
        \grd@target
        \pgfmathsetmacro{\grd@xb}{\the\pgf@x/1cm}
        \pgfmathsetmacro{\grd@yb}{\the\pgf@y/1cm}
        \pgfmathsetmacro{\grd@xc}{\grd@xa + \pgfkeysvalueof{/tikz/grid with coordinates/major step}}
        \pgfmathsetmacro{\grd@yc}{\grd@ya + \pgfkeysvalueof{/tikz/grid with coordinates/major step}}
        \foreach \x in {\grd@xa,\grd@xc,...,\grd@xb}
        \node[anchor=north] at (\x,\grd@ya) {\pgfmathprintnumber{\x}};
        \foreach \y in {\grd@ya,\grd@yc,...,\grd@yb}
        \node[anchor=east] at (\grd@xa,\y) {\pgfmathprintnumber{\y}};
      }
    }
  },
  minor help lines/.style={
    help lines,
    step=\pgfkeysvalueof{/tikz/grid with coordinates/minor step}
  },
  major help lines/.style={
    help lines,
    line width=\pgfkeysvalueof{/tikz/grid with coordinates/major line width},
    step=\pgfkeysvalueof{/tikz/grid with coordinates/major step}
  },
  grid with coordinates/.cd,
  minor step/.initial=.2,
  major step/.initial=1,
  major line width/.initial=2pt,
}
\makeatother

\usepackage{thmtools}
\usepackage{thm-restate}
\usepackage{etoolbox}
\makeatletter
\def\problem@s{}
\newcounter{problems@cnt}

\newcommand{\allproblems}{\problem@s}
\makeatother